\documentclass[lettersize,journal]{IEEEtran}
\usepackage{amsmath,amsfonts}
\usepackage{algorithmic}
\usepackage{algorithm}
\usepackage{array}
\usepackage[caption=false,font=normalsize,labelfont=sf,textfont=sf]{subfig}
\usepackage{textcomp}
\usepackage{stfloats}
\usepackage{url}
\usepackage{verbatim}
\usepackage{graphicx}
\usepackage{cite}
\newcommand {\Define} {\stackrel {\Delta} {=}  }
\newcommand{\mya}{\mathrel{\overset{\makebox[0pt]{{\tiny(a)}}}{=}}}
\newcommand{\myb}{\mathrel{\overset{\makebox[0pt]{{\tiny(b)}}}{=}}}
\newcommand{\myc}{\mathrel{\overset{\makebox[0pt]{{\tiny(c)}}}{=}}}

\newtheorem{theorem}{Theorem}
\newtheorem{lemma}{Lemma}

\begin{document}

\title{Inter-frame Channel Prediction for Zak-OTFS}
\author{\IEEEauthorblockN{Muhammad Ubadah\IEEEauthorrefmark{1} and Saif Khan Mohammed\IEEEauthorrefmark{1}\IEEEauthorrefmark{2}
}
\\
\IEEEauthorblockA{\IEEEauthorrefmark{1}Department of Electrical Engineering, Indian Institute of Technology Delhi, India}\\
\IEEEauthorblockA{\IEEEauthorrefmark{2}Cohere Technologies Inc., San Jose, CA, USA}\\
\thanks{M. Ubadah and S. K. Mohammed are with the Department of Electrical Engineering, Indian Institute of Technology Delhi, India (E-mail: eez198356@ee.iitd.ac.in, saifkmohammed@gmail.com). S. K. Mohammed is currently with Cohere Technologies Inc., CA, USA, on extra-ordinary leave from I.I.T. Delhi.}
}




\maketitle

\begin{abstract}
Zak-Orthogonal Time Frequency Space (OTFS) modulation is known to be robust to Doppler spread in high mobility scenarios when compared to Orthogonal Frequency Division Multiplexing (OFDM). This is due to the fact that the channel response to a Zak-OTFS carrier within a frame can be accurately estimated from the channel response to another carrier within the same frame. However, an important open problem and question is whether inter-frame channel prediction is possible with Zak-OTFS, i.e., is it possible to accurately predict the channel response to a Zak-OTFS carrier in a frame based on knowledge of the channel response to some Zak-OTFS carrier in \emph{another} frame (i.e., not the same frame).

In this paper we show that indeed inter-frame channel prediction is possible. We show that the effective DD domain channel filter coefficients vary in a deterministic manner as we move from current to future frames in time and frequency. We also show that the subspace spanned by channel filter coefficients of consecutive frames in time/frequency is invariant to discrete shifts in time and frequency. We exploit the deterministic variation and subspace invariance to propose a novel deterministic ESPIRIT-type method which uses the effective DD domain channel filter taps/coefficients estimated in training frames (i.e., current/past frames in time and frequency having both pilot and data carriers) to predict the  effective DD domain channel filter for frames which are several tens of frames in future and several tens of frames away in frequency.

Exhaustive numerical simulations for the Vehicular-A channel reveal that even for a high Doppler spread of $2$ kHz, the normalized prediction error for a frame $120$ ms into future and $432$ MHz away is roughly $-14.2$ dB when the pilot power to noise power ratio is $15$ dB in the training frame. An application of the proposed prediction is to reduce the pilot overhead since the prediction frames (for which channel is predicted) do not need dedicated pilot. Simulation of coded Zak-OTFS frames reveal that for a Doppler spread of $2$ kHz and total transmit power to noise power ratio of $15$ dB, the effective spectral efficiency (SE) for a prediction frame $60$ ms in future and $216$ MHz away is in fact $30$ percent higher than the SE of traditional frames which have both  pilot and data.
\end{abstract}

\section{Introduction}
Next generation communication systems are expected to deploy Artificial Intelligence (AI) and Machine Learning (ML) to achieve significant improvements in network performance (spectral and energy efficiency, latency, reliability and scalability) \cite{intro1,intro2}. 
These improvements are achieved through integration of AI into the protocol stack thereby enabling smart resource allocation, predictive traffic management, dynamic beamforming,  interference management, smart sensing\cite{QW2026,SM2025}.
Wireless channels in 6G and beyond are expected to be highly time and frequency selective due to scenarios where the channel delay and Doppler spread can be high \cite{IMT2030}. 
Some examples include, satellite based non-terrestrial networks, high speed train, Vehicle-to-everything communication, (V2X), Unmanned Aerial Vehicle (UAV) communication, High-Altitude-Platform-Stations (HAPS) communication, Aircraft to Ground communication.

Fast time- and frequency domain variations of the communication channel gain, result in higher density of pilots in time and frequency which reduces network throughput in Orthogonal Frequency Division Multiplexing (OFDM) based systems. This also increases the frequency and overhead of channel feedback between the network and the user terminals. Fast time- and frequency variations also degrades reliability due to fading in time and frequency. Higher channel Doppler spreads result in inter-carrier interference (ICI) between OFDM carriers which also degrades throughput and reliability. Outdated and inaccurate channel state information (CSI) at the network side results in sub-optimal resource management, beamforming and interference management.
Therefore, AI-based performance improvements are only feasible if the network can predict the channel state information (CSI) for future frames in time and also for frames which are separated in frequency (for example in different frequency bands).

Most CSI prediction methods can be classified as either being based on Autoregression (AR) or based on AI models (Deep Neural Networks (DNN), Long Short-term memory (LSTM) networks, transformer models) \cite{XWei24, HHou26,Bchong26}. Prediction based on AI models suffers from practical infeasibility as they need very large amount of historical CSI to train their models. AI based prediction also has very high computational complexity due to which predicted CSI becomes obsolete by the time network is ready to use it. Also, in high mobility scenarios the prediction horizon of AI based prediction is small, i.e., the prediction accuracy degrades severely even for frames which are only a few milliseconds in future \cite{HHou26, Bchong26}.
AR based CSI prediction methods have lower complexity than AI based methods and also do not require as much historical CSI. However, their performance deteriorates in non-line-of-sight high-mobility channels and therefore their prediction horizon is also limited \cite{Paiva25,Zchen,Wang23}.

Recently proposed Zak - Orthogonal Time Frequency Space (OTFS) modulation is
based on signaling in the delay-Doppler (DD) domain where information is carried by narrow quasi-periodic pulses in the DD domain \cite{zakotfs1,zakotfs2, otfsbook}.
A Zak-OTFS frame consists of a finite number of quasi-periodic DD domain pulses located on a DD domain lattice/grid, with each DD pulse carrying a different information symbol.
The time-domain (TD) realization of a Zak-OTFS DD pulse is a periodic pulse train modulated by a sinusoid/tone (referred to as a pulsone).
All pulsones in a Zak-OTFS frame span the same
time-frequency (TF) interval, and pulsones in different Zak-OTFS frames occupy non-overlapping TF intervals.

Zak-OTFS is known to be more robust to channel delay and Doppler spread when compared to OFDM.
This is because, even in doubly-spread channels, the channel response to a Zak-OTFS pulsone can be accurately estimated from the known channel response to another pulsone in the same frame. This is not possible with sinusoidal carriers in OFDM, since the channel response to a subcarrier cannot be accurately estimated from the response to another subcarrier. This is why in OFDM we need to repeat pilot subcarriers in both time and frequency, with the repetition frequency increasing with increasing channel delay and Doppler spread which thereby increases channel estimation overhead.
On the other hand, in a Zak-OTFS frame, it suffices to dedicate a single pulsone carrier as a pilot carrier since the response to this pilot carrier can be used to estimate the response to other pulsone carriers which carry data. Also, the
Zak-OTFS I/O relation is non-selective even in doubly-spread channels, i.e., the energy of the channel response to a Zak-OTFS pulsone is equal to the energy of the response to any other pulsone \cite{zakotfs1, zakotfs2, otfsbook}. We describe the system model for a doubly-spread channel in Section \ref{secsysmodel} and multi-frame Zak-OTFS modulation in Section \ref{seczakotfs}.

Although, it is known how to estimate the channel response to a Zak-OTFS pulsone within a frame from the response to a pulsone within the same frame, it is an open problem on how to extend this feature of Zak-OTFS modulation across multiple-frames. In other words, to the best of our knowledge there is no known work in literature which allows us to estimate/predict the channel response to a Zak-OTFS pulsone in one frame from the known channel response to a Zak-OTFS pulsone in another frame. Within a frame, the channel response to a pulsone is described in terms of DD domain twisted convolution between the DD representation of the pulsone and an effective DD domain channel filter for that frame. Therefore, this open problem can be solved if we have a method to predict the effective DD domain channel filter in a frame from the knowledge of the channel filter in another frame. 

In this paper, we solve this open problem
by proposing a method which can accurately predict/estimate the effective DD domain channel filter in a future frame in another frequency band from the knowledge of effective DD domain channel filter acquired in few training frames. We exploit the fact that the
effective DD domain channel filter taps/coefficients are in one-to-one correspondence with the DD signature/profile of the underlying physical multi-path channel which varies slowly in time/frequency.\footnote{\footnotesize{The physical delay and Doppler shift of each channel path depends on the geometry of reflectors around the transceiver, the relative speed and orientation of the transceiver, which change at a slower time-scale when compared to the time-coherence of sinusoidal carriers in OFDM.}}


In Section \ref{seciorel} we show the novel result that the coefficient for each DD tap of the effective DD domain channel filter varies from one frame to another in a deterministic manner. To be precise, for the contribution of the $i$-th channel path to each DD tap, an extra phase of $2 \pi \nu_i T'$ gets added as we move by one frame into the future and an extra phase of $2 \pi \tau_i B'$ gets added if we move backward in frequency by one frame. Here $T'$ and $B'$ denote the time duration and bandwidth of each frame respectively. Also, $\tau_i$ and $\nu_i$ denote the delay and Doppler shift of the $i$-th channel path.

We also make the novel observation that the
subspace spanned by the DD domain channel filter coefficients of consecutive frames in time/frequency is invariant of discrete shifts in time and frequency.
This subspace invariance and the deterministic variation of the DD domain filter coefficients allows us to propose an ESPIRIT-type\footnote{\footnotesize{ESPIRIT stands for Estimation of Signal Parameters Via Rotational Invariance Techniques, and is a well known technique used for time-series analysis, direction-of-arrival estimation \cite{espirit_paper}.}} method to accurately estimate the phase gain parameters $2 \pi \nu_i T'$ and $2 \pi \tau_i B'$ (modulo $2 \pi $) for all channel paths, given the estimates of the channel filter taps for the current and past $Q$ frames in time and $Q$ frames in frequency (these $(2Q+1)$ frames are referred to as training frames).
The training frames carry both pilot and data.
The prediction frames (i.e., frames for which we predict the DD channel filter) need not carry pilot. The proposed prediction method is described in detail in Section \ref{propmethod}. The overall prediction complexity is only linear in the number of DD taps in the channel filter and cubic in the number of training frames.

Through exhaustive numerical simulations for the six-path Vehicular-A channel in Section \ref{secsim}, we show the effectiveness of the proposed method in accurately predicting the effective DD domain channel filter for Zak-OTFS frames up to roughly $100$ ms in future and several hundred MHz away in frequency.
Numerical simulations show that the derived estimates of the phase parameters are close to the Cramer Rao Lower Bound (CRLB) when the number of training frames is sufficiently larger than the number of channel paths.
Even in high mobility scenarios (channel Doppler spread of $2$ kHz), the proposed prediction method (with $Q=30$ and a pilot to noise ratio of $15$ dB in the training frames) achieves a normalized prediction error of $-14.2$ dB for a prediction frame which is $120$ ms in future and $432$ MHz away in frequency. Simulation of LDPC coded Zak-OTFS frames shows significant improvement in the spectral efficiency of prediction frames when compared to traditional frames where DD channel filter is estimated using dedicated pilot. This is because, traditional pilot based frames need to also dedicate guard DD  pulse carriers around the pilot DD pulse carrier so as to avoid interference between data DD pulse carriers and the pilot carrier. Since guard DD pulse carrier do not carry data, they are an overhead which reduces the effective spectral efficiency (SE) of traditional frames. In comparison, prediction frames do not need pilot and therefore no guard DD carriers are required, resulting in no overhead. Through simulations we see that even for a high channel Doppler spread of $2$ kHz, a prediction frame $60$ ms in future and $216$ MHz away achieves $30$ percent higher SE than a traditional frame for a total transmit power to noise ratio of $18$ dB.   

 The proposed prediction method has several applications:
 \begin{itemize}
\item The frames where channel is predicted, need not carry pilot which improves SE significantly as the pilot power can be allocated to data and also more number of data symbols can be transmitted in each frame. In the absence of pilot, the PAPR also reduces. The improvement in SE due to reduction in pilot overhead is even more significant for MIMO systems.
\item In the downlink of Time Division Duplexed (TDD) systems, the transmitter can predict the downlink channel based on the acquired estimates in the uplink and can use it for downlink precoding which can significantly reduce the receiver equalization complexity. Also, unlike traditional systems, there is no overhead of CSI feedback from the receiver to the transmitter.
\item Channel prediction at the network side (based on uplink packets from the user terminals), can be helpful in optimizing resource allocation.
\item In traditional TDD systems, channel reciprocity allows us to design the downlink beamforming based on pilots received in the uplink. However, in Frequency Division Duplexed (FDD) systems this is not possible and CSI estimates have to be fed back from the user terminal back to the network.
This CSI feedback overhead can be high in MIMO and Massive MIMO systems. Since the proposed prediction method can also predict the DD domain channel coefficients for frames transmitted in different frequency bands several hundreds of MHz apart, no/minimal CSI feedback is required which improves the overall throughput significantly.
 \end{itemize}
    
\section{System model}
\label{secsysmodel}
Consider the physical DD spreading function to be
\begin{eqnarray}
\label{eqn1}
    h_{\mbox{\scriptsize{phy}}}(\tau, \nu) & = & \sum\limits_{i=1}^P h_i \, \delta(\tau - \tau_i) \, \delta(\nu  - \nu_i),
\end{eqnarray}where $P$ is the number of physical channel paths and $h_i, \tau_i, \nu_i$ denote the complex gain, delay and Doppler shift of the $i$-th path. The input time-domain (TD) signal $s(t)$ and output $r(t)$ are related by \cite{Bello}
\begin{eqnarray}
\label{eqn2}
    r(t) & \hspace{-3mm} = &  \hspace{-3mm} \iint h_{\mbox{\scriptsize{phy}}}(\tau, \nu) \, s(t - \tau) \, e^{j 2 \pi \nu (t - \tau)} \, d\tau \, d\nu \, + \, z(t)
\end{eqnarray}where $z(t)$ is AWGN. 

\section{Multi-frame Zak-OTFS modulation}
\label{seczakotfs}
We consider multiple Zak-OTFS frames, each of time duration $T'$ and bandwidth $B'$. The $(n,m)$-th Zak-OTFS frame is approximately localized to the TF interval
\begin{eqnarray}
    {\mathcal I}_{n,m} & \Define & \left\{ (t,f) \, {\Big |} \, \vert t - nT' \vert < \frac{T'}{2} \,,\, \vert f - mB' \vert < \frac{B'}{2}  \right\},
\end{eqnarray}i.e., approximately localized to the
time interval $\left[ nT' - T'/2 \,,\, nT' + T'/2 \right]$ and frequency interval $\left[ mB' - B'/2 \,,\, mB' + B'/2 \right]$, $(n,m) \in {\mathbb Z}$.

Zak-OTFS modulation is parameterized by the delay parameters, delay period $\tau_p > 0$ and Doppler period $\nu_p = 1/\tau_p$ and the DD domain information lattice  \cite{zakotfs1, zakotfs2}
\begin{eqnarray}
\Lambda & = & \left\{ \left(\frac{k}{B} \,,\, \frac{l}{T} \right) \, \vert \, k,l \in {\mathbb Z} \right\},
\end{eqnarray}where $1/B$ and $1/T$ denote the spacing between the lattice points along the delay and Doppler axis respectively.
Also, let $M \Define B \tau_p$ and $N \Define T \nu_p$ be integers. $x_{n,m}[k,l]$, $k=0,1,\cdots, M-1$, $l=0,1,\cdots, N-1$
denotes the $MN$ information symbols to be transmitted on the $(n,m)$-th frame.
These information symbols are embedded into a discrete DD domain signal \cite{zakotfs1, zakotfs2, otfsbook}
\begin{eqnarray}
    x_{n,m,dd}[k,l] & \hspace{-3mm} \Define & \hspace{-3mm} \sum\limits_{n',m' \in {\mathbb Z}} \sum\limits_{k'=0}^{M-1} \sum\limits_{l'=0}^{N-1} {\Big (} x_{n,m}[k', l'] \, e^{j \frac{2 \pi n' l'}{N}} \nonumber \\
    & & \hspace{4mm} \, \delta[k - k' - n'M] \, \delta[l - l' - m'N] \, {\Big )},
\end{eqnarray}which is quasi-periodic with periods $M$ and $N$ along the delay and Doppler axis respectively, i.e.
\begin{eqnarray}
x_{n,m,dd}[k+n'M,l+m'N] & = & e^{j \frac{2 \pi n' l}{N}} \, x_{n,m,dd}[k,l]
\end{eqnarray}for all $n',m' \in {\mathbb Z}$. This discrete DD domain information signal is then lifted to the points of the information lattice $\Lambda$ resulting in the continuous DD domain information signal
\begin{eqnarray}
    x_{\mbox{\scriptsize{n,m,dd}}}(\tau, \nu) & \hspace{-3mm} = & \hspace{-4mm} \sum\limits_{k, l \in {\mathbb Z}}  \hspace{-2mm} x_{\mbox{\scriptsize{n,m,dd}}}[k,l] \, \delta\left( \tau - \frac{k \tau_p}{M}\right) \delta\left( \nu - \frac{l \nu_p}{N}\right)
\end{eqnarray}which is a continuous DD domain quasi-periodic function, with periods $\tau_p$ and $\nu_p$ along the delay and Doppler axis respectively, i.e.
\begin{eqnarray}
    x_{\mbox{\scriptsize{n,m,dd}}}(\tau + n' \tau_p, \nu + m' \nu_p) & = & e^{j 2 \pi n' \nu \tau_p} \,  x_{\mbox{\scriptsize{n,m,dd}}}(\tau , \nu )
\end{eqnarray}for all $n',m' \in {\mathbb Z}$. This is then filtered (pulse shaping) with a DD filter $w_{n,m,tx}(\tau, \nu)$ so that the transmitted TD signal is approximately localized to the TF interval ${\mathcal I}_{n,m}$. 
The filtered information signal is given by
\begin{eqnarray}
\label{eqn8}
    x_{n,m,dd}^{w}(\tau , \nu ) & = & w_{n,m,tx}(\tau, \nu) \, *_{\sigma} x_{n,m,dd}(\tau , \nu ),
\end{eqnarray}where $*_{\sigma}$ denotes the twisted convolution operation defined for any two DD domain functions $a(\tau, \nu)$ and $b(\tau, \nu)$ as
\begin{eqnarray}
\label{eqneqn9}
c(\tau, \nu) & = & a(\tau, \nu) *_{\sigma} b(\tau, \nu) \nonumber \\
& & \hspace{-18mm} = \iint a(\tau', \nu') \, b(\tau - \tau', \nu - \nu') \, e^{j 2 \pi \nu' (\tau - \tau')} \, d\tau' \, d\nu'.
\end{eqnarray}
The transmitted TD signal for the $(n,m)$-th frame is the TD realization of $x_{n,m,dd}^{w}(\tau , \nu )$, i.e., its inverse Zak-transform (${\mathcal Z}_t^{-1}$) given by \cite{Zak67, Janssen88, DerivOTFS}
\begin{eqnarray}
\label{eqn10}
    s_{n,m}(t) & = & {\mathcal Z}_t^{-1}\left(  x_{n,m,dd}^{w}(\tau , \nu ) \right) \nonumber \\
    & = & \sqrt{\tau_p} \int\limits_{0}^{\nu_p} x_{n,m,dd}^{w}(t , \nu ) \, d\nu.
\end{eqnarray}The next theorem gives the expression for the pulse shaping filter for the $(n,m)$-th frame.
\begin{theorem}
\label{thm1}
The pulse shaping filter
\begin{eqnarray}
\label{eqn1212}
w_{n,m,tx}(\tau,\nu) & \Define & \delta(\tau - nT') \delta(\nu - mB') *_{\sigma} w_{0,0,tx}(\tau, \nu), \nonumber \\
&  & \hspace{-17mm} = e^{j2\pi m B' (\tau - nT')} \, w_{0,0,tx}(\tau - nT', \nu - mB'), 
\end{eqnarray}guarantees that
the signal $s_{n,m}(t)$ transmitted for the $(n,m)$-th frame is approximately localized to the TF interval ${\mathcal I}_{n,m}$. Note that $w_{0,0,tx}(\tau, \nu)$ is the transmit pulse shaping filter for the $(0,0)$-th frame.
\end{theorem}
\begin{IEEEproof}
See Appendix \ref{appen_thm1}.

\end{IEEEproof}
The total signal transmitted for all frames is given by
\begin{eqnarray}
\label{eqneqn16}
    s(t) & = & \sum\limits_{n,m \in {\mathbb Z}} s_{n,m}(t).
\end{eqnarray}

Next, we discuss receiver signal processing for the $(n,m)$-th frame.
At the receiver, the Zak-transform of the received signal $r(t)$ gives its DD domain realization \cite{Zak67, Janssen88,DerivOTFS}
\begin{eqnarray}
\label{eqn11}
    y_{\mbox{\scriptsize{dd}}}(\tau , \nu ) & = & {\mathcal Z}_t \left(  r(t)\right) \nonumber \\
    & & \hspace{-10mm} = \sqrt{\tau_p} \sum\limits_{n \in {\mathbb Z}} r(\tau + n \tau_p) \, e^{-j 2 \pi \nu n \tau_p}.
\end{eqnarray}
This is followed by matched filtering. For an arbitrary transmit pulse shaping filter $w_{tx}(\tau, \nu)$ the corresponding matched filter at the receiver is given by $w_{rx}(\tau, \nu) = w_{tx}^*(-\tau, -\nu) \, e^{j 2 \pi \nu \tau}$ \cite{Hanly24}. Therefore the
matched filter for the $(n,m)$-th frame is
\begin{eqnarray}
\label{eqne1818}
    w_{n,m,rx}(\tau, \nu) & \Define & w_{n,m,tx}^*(-\tau, -\nu) \, e^{j2\pi\tau\nu}.
\end{eqnarray}The output of the matched filter is given by
\begin{eqnarray}
\label{eqnref19}
y_{n,m,dd}^{w}(\tau , \nu ) & = & w_{n,m,rx}(\tau, \nu) \, *_{\sigma} \, y_{\mbox{\scriptsize{dd}}}(\tau , \nu ).
\end{eqnarray}Sampling $y_{n,m,dd}^{w}(\tau , \nu ) $ on the information lattice $\Lambda$ gives the received discrete DD domain signal
\begin{eqnarray}
y_{n,m,dd}[k,l] & = & y_{n,m,dd}^{w}\left( \tau = \frac{k}{B} , \nu = \frac{l}{T} \right),
\end{eqnarray}$k,l \in {\mathbb Z}.$
Fig.~\ref{figmultiframezak} shows the transceiver signal processing in multi-frame Zak-OTFS.
    \begin{figure}
     \centering
\includegraphics[width=9.0cm,height=3.6cm]{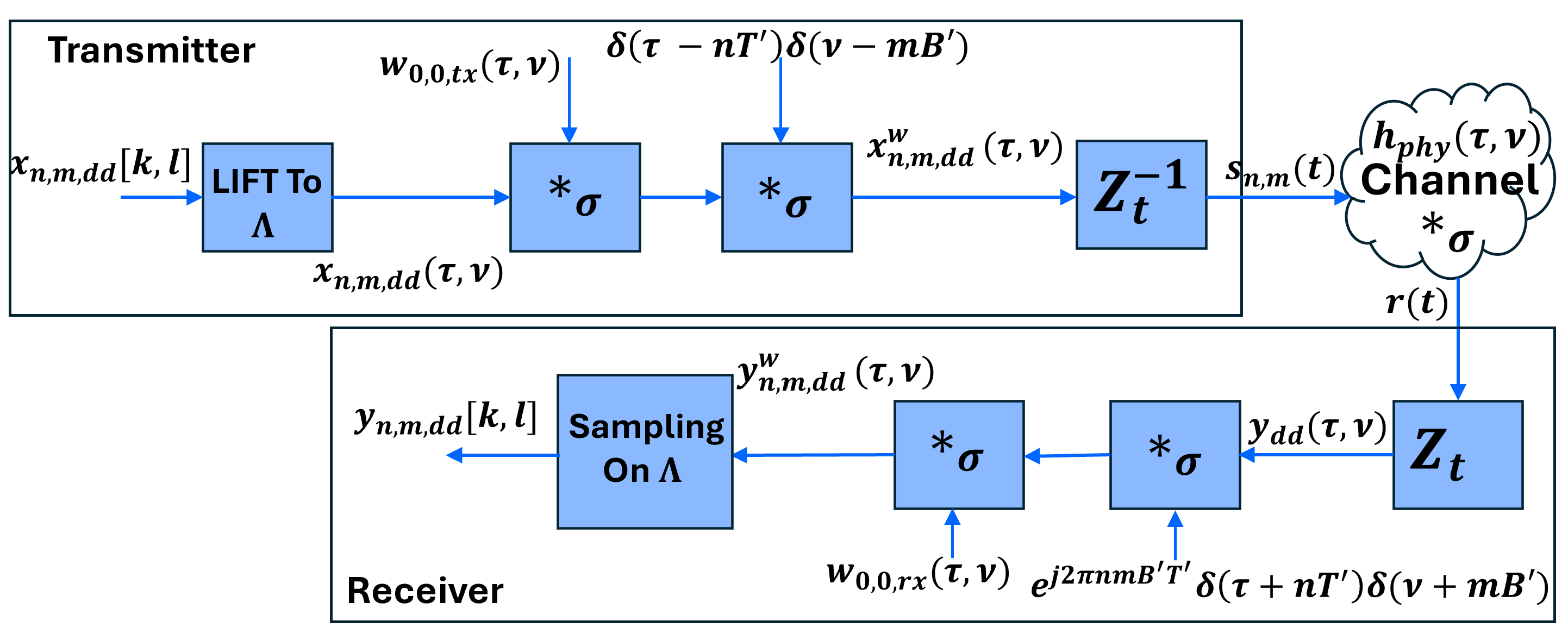}
        \hspace{-6mm}
        \caption{Transceiver signal processing in multi-frame Zak-OTFS.}
        \label{figmultiframezak}
        \vspace{-3mm}
    \end{figure}

\section{I/O relation of multi-frame Zak-OTFS modulation}
\label{seciorel}
The following theorem gives the input-output (I/O) relation for the $(n,m)$-th frame.
\begin{theorem}
\label{thm2}
The I/O relation for the $(n,m)$-th frame is given by
\begin{eqnarray}
\label{thm2eqn1}
y_{n,m,dd}[k,l] & = & h_{n,m}[k,l] *_{\sigma} x_{n,m,dd}[k,l] \nonumber \\
& & \hspace{-7mm} + \underbrace{\hspace{-4mm}\sum\limits_{(p,q) \ne (n,m)} \hspace{-4mm} h_{n,m,p,q}[k,l] *_{\sigma}
 x_{p,q,dd}[k,l]}_{\text{Inter Frame Interference (IFI)}} \nonumber \\
 & & + z_{n,m,dd}[k,l],
\end{eqnarray}where twisted convolution ($*_{\sigma}$) between two discrete DD domain functions $a[k,l]$ and $b[k,l]$ is given by
\begin{eqnarray}
    c[k,l] & = & a[k,l] *_{\sigma} b[k,l] \nonumber \\
    & & \hspace{-13mm} = \sum\limits_{k', l' \in {\mathbb Z}}  a[k',l'] \, b[k - k', l - l'] \, e^{j 2 \pi \frac{l' (k - k')}{MN}}.
\end{eqnarray}
In (\ref{thm2eqn1}), the effective DD domain filter $h_{n,m}[k,l]$ for the $(n,m)$-th frame is given by
\begin{eqnarray}
\label{thm2eqn3}
    h_{n,m}[k,l] & = & h_{n,m,n,m}\left( \tau = \frac{k}{B} \,,\, \nu = \frac{l}{T} \right), \nonumber \\
     h_{n,m,p,q}[k,l] & = & h_{n,m,p,q}\left( \tau = \frac{k}{B} \,,\, \nu = \frac{l}{T} \right),
\end{eqnarray}where
\begin{eqnarray}
\label{eqneqn24}
    h_{n,m,p,q}(\tau, \nu) & \Define & w_{n,m,rx}(\tau,\nu) *_{\sigma} h_{\text{phy}}(\tau, \nu) *_{\sigma} w_{p,q,tx}(\tau,\nu). \nonumber \\
\end{eqnarray}Also, $z_{n,m,dd}[k,l]$ are the DD domain AWGN samples obtained after sampling the filtered AWGN signal $w_{n,m,rx}(\tau, \nu) *_{\sigma} z_{dd}(\tau, \nu)$ on $\Lambda$ (here $z_{dd}(\tau,\nu)$ is the DD domain representation of AWGN $z(t)$).
\end{theorem}
\begin{IEEEproof}
See Appendix \ref{appen_thm2}.
\end{IEEEproof}
The following theorem expresses the effective DD channel filter $h_{n,m}[k,l]$ in terms of the parameters of the physical channel $h_{\text{phy}}(\tau, \nu)$.
\begin{theorem}
\label{thm3}
The effective discrete DD domain channel filter $h_{n,m}[k,l]$ can be decomposed as
\begin{eqnarray}
\label{eqnthm3}
    h_{n,m}[k,l] & = & \sum\limits_{i=1}^P \alpha_i^n \, \beta_i^m \, A_i[k,l], \nonumber \\
    \alpha_i & \Define & e^{j2\pi \nu_i T'} \,,\, \beta_i  \Define e^{-j2\pi \tau_i B'}, \nonumber \\
    A_i[k,l] & \Define & A_i\left(\tau = \frac{k}{B} \,,\, \nu = \frac{l}{T} \right), \nonumber \\
    A_i(\tau, \nu) & \Define & h_i {\Big (} w_{0,0,rx}(\tau, \nu) *_{\sigma} (\delta(\tau - \tau_i) \delta(\nu - \nu_i)) \nonumber \\
    & & \hspace{4mm} *_{\sigma} w_{0,0,tx}(\tau, \nu) {\Big )}. \nonumber \\
\end{eqnarray}
\end{theorem}
\begin{IEEEproof}
See Appendix \ref{appen_thm3}.
\end{IEEEproof}
An important observation from (\ref{eqnthm3}) is that $A_i[k,l]$ does not depend on $(n,m)$. For the $(n,m) = (0,0)$-th frame, from (\ref{eqnthm3}) it follows that
\begin{eqnarray}
    h_{0,0}[k,l] & = & \sum\limits_{i=1}^P A_i[k,l].
\end{eqnarray}In other words, $A_i[k,l]$ is the contribution of the $i$-th channel path to the DD domain effective channel filter for the $(n,m) = (0,0)$-th frame. From the expression for $h_{n,m}[k,l]$ in (\ref{eqnthm3}) it is also clear that as move forward in time to the next frame (time advancement by $T'$) the $i$-th path introduces the phase factor $\alpha_i = e^{j 2 \pi \nu_i T'}$ due to the Doppler shift $\nu_i$. Similarly, when we move forward to the next frame in frequency (i.e. moving by $B'$), the $i$-th path introduces the phase factor $\beta_i = e^{-j 2\pi \tau_i B'}$. We exploit this behavior to propose a method which uses several training frames in time and in frequency for which the effective DD domain channel filter is known, to predict the effective channel filter for any $(n,m)$-th frame.

\section{Proposed Prediction of the effective DD domain channel filter}
\label{propmethod}
The methodology/philosophy of the proposed method is as follows.
\begin{itemize}
    \item \textbf{Step} ${\bf 1}$ - Acquire channel estimates in training frames: The input to the prediction method are the estimates of the DD domain effective channel filter acquired in $K$ ``training frames" with $(n,m)$ indices in the set ${\mathcal F} = \{ (n_1, m_1), \cdots, (n_K, m_K) \}$. Also, let ${\mathcal S} = \{ (k_1,l_1), \cdots, (k_{N_t}, l_{N_t}) \}$ denote a set of $N_t$ DD taps $(k,l)$ where $h_{n,m}[k,l]$ has significant energy. We assume that the physical channel $h_{\text{phy}}(\tau, \nu)$ is stationary and does not vary with time.
    \item \textbf{Step} ${\bf 2}$ - Estimate channel parameters based on channel estimates acquired in training frames: Based on the channel estimates $\widehat{h}_{n,m}[k,l]$, $(n,m) \in {\mathcal F}$, $(k,l) \in {\mathcal S}$ from the training frames, we estimate the parameters $P$, $\alpha_i, \beta_i$, $i=1,2,\cdots, P$, $A_i[k,l]$, $i=1,2,\cdots,P$, $(k,l) \in {\mathcal S}$.
    \item \textbf{Step} ${\bf 3}$ - Use estimated channel parameters $\widehat{P}$, $\widehat{\alpha}_i, \widehat{\beta}_i$, $\widehat{A}_i[k,l]$, to predict the effective DD domain channel filter in any arbitrary $(n,m)$-th frame, i.e.
    \begin{eqnarray}
        \label{est926649}
        \widehat{h}_{n,m}[k,l] & = & \sum\limits_{i=1}^{\widehat{P}} \left(\widehat{\alpha}_i\right)^n  \, \left(\widehat{\beta}_i\right)^m  \, \widehat{A}_i[k,l],
    \end{eqnarray}for all $(k,l) \in {\mathcal S}$.
\end{itemize}

In this work, we consider a particular set ${\mathcal F}$ of training frames consisting of the current frame $(n,m) = (0,0)$ and the past $Q$ frames in time, i.e., $\{ (-Q,0), (-(Q-1), 0), \cdots, (-1,0)\}$, and $Q$ frames in frequency, i.e., $\{ (0,-Q), (0, -(Q-1)), \cdots, (0,-1)\}$,
\begin{eqnarray}
\label{eqnframes}
    {\mathcal F} & = & \{ (-Q,0), (-(Q-1), 0), \cdots, (-1,0), (0,0), \nonumber \\
    & & \hspace{4mm} (0,-Q), (0, -(Q-1)), \cdots, (0, -1) \}.
\end{eqnarray}In the following we firstly  enumerate the assumptions and their validity
in Section \ref{subsecassump} followed by a detailed description of the proposed prediction method in Section \ref{subsecpredict} for the training frames described above in (\ref{eqnframes}). We would like to mention here that the described method is not limited to the training frames in (\ref{eqnframes}) but is applicable to many other possible set of training frames. In frames where we predict the effective DD domain channel filter, all $MN$ pulsones/DD carriers carry data symbols (as they need not carry pilot) and are called as prediction frames. Fig.~\ref{figmultiframe}
illustrates training and prediction frames for the choice of ${\mathcal F}$
in (\ref{eqnframes}). Training frames carry both pilot and data and the
estimation of the effective discrete DD domain channel filter from the received pilot response in these frames is discussed in Section \ref{subsectrainest}.
    \begin{figure}
     \centering
\includegraphics[width=9.0cm,height=5.6cm]{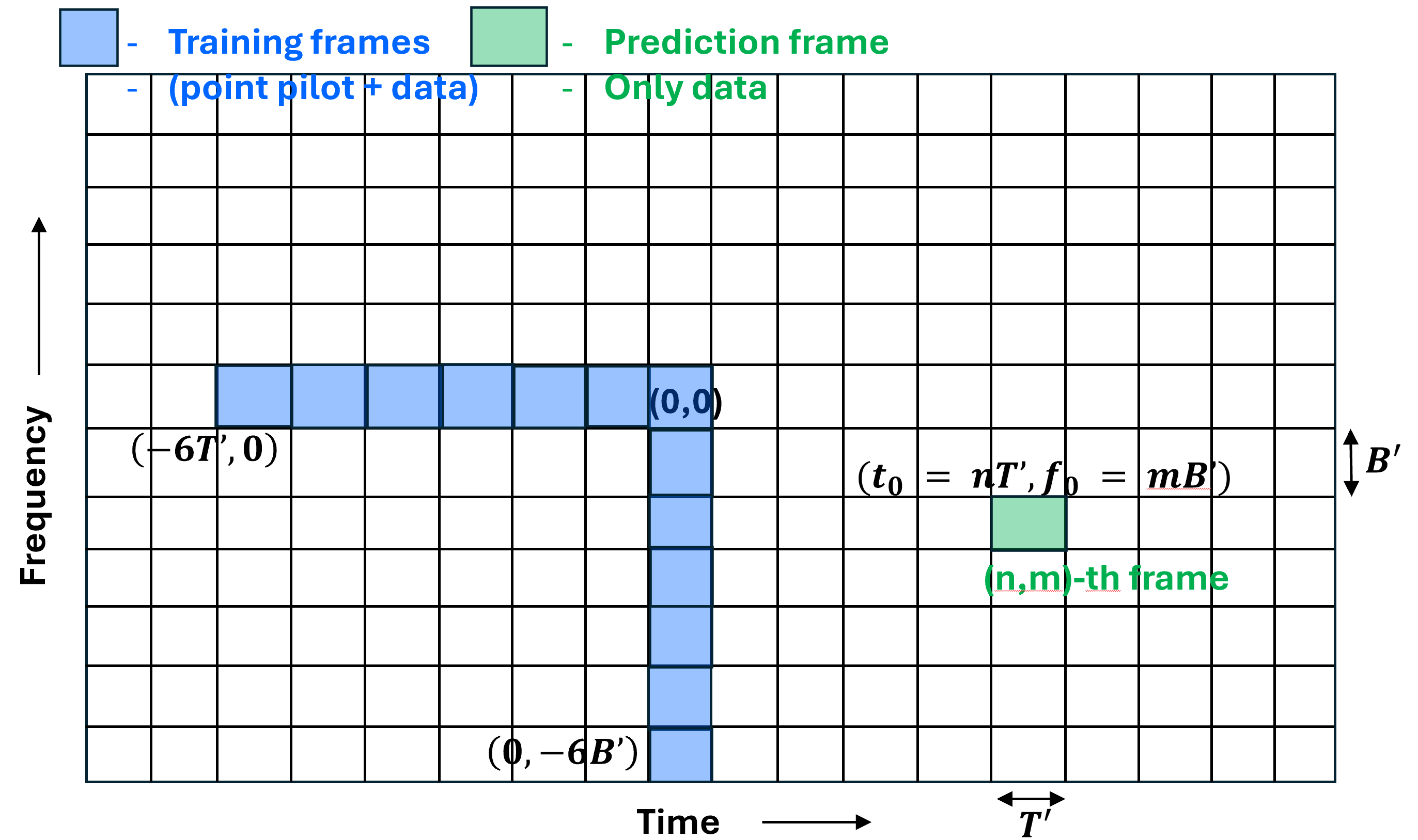}
        \hspace{-6mm}
        \caption{Multi-frame Zak-OTFS and inter-frame prediction.}
        \label{figmultiframe}
        \vspace{-3mm}
    \end{figure}

\subsection{Assumptions and their validity}
\label{subsecassump}
In this section we enumerate the practical assumptions for the proposed
prediction method.
Let ${\bf A} \in {\mathbb C}^{N_t \times P}$ be the matrix whose $i$-th column is $A_i[k,l], (k,l) \in {\mathcal S}$.
\begin{eqnarray}
\label{eqnadef}
    {\bf A} &\hspace{-3.5mm} \Define & \hspace{-3.5mm} \begin{bmatrix}
        A_1[k_1, l_1] & A_2[k_1, l_1] & \dots & A_P[k_1, l_1] \\
        A_1[k_2, l_2] & A_2[k_2, l_2] & \dots & A_P[k_2, l_2] \\
        \vdots & \vdots & \vdots & \vdots  \\
         A_1[k_{N_t}, l_{N_t}] & A_2[k_{N_t}, l_{N_t}] & \dots & A_P[k_{N_t}, l_{N_t}].
    \end{bmatrix}
\end{eqnarray}We refer to ${\bf A}$ as the DD signature matrix, since its columns describe the delay Doppler profile of each channel path.
With sufficiently high $B$ and $T$ so that the channel paths are almost resolvable in the DD domain (i.e., each DD bin receives most if its energy from a single or at most a few channel paths), and fractional channel delay and Doppler shifts\footnote{\footnotesize{Channel delay and Doppler shifts are said to be fractional if they are non-integer multiples of $1/B$ and $1/T$ respectively.}}, it has been observed that the number of significant DD taps is greater than the number of channel paths $P$
\begin{eqnarray}
\label{eqnntp}
    N_t & \geq & P.
\end{eqnarray}A possible way of choosing $N_t$ and ${\mathcal S}$ is discussed in Section \ref{subsecnt}. Through exhaustive simulations we have also observed that if the channel paths are distinct, i.e. for all $i \ne i'$, $i, i' = 1,2,\cdots, P$
\begin{eqnarray}
\label{eqndistinct}
    (\tau_i, \nu_i) & \ne & (\tau_{i'}, \nu_{i'}),
\end{eqnarray}then the matrix ${\bf A}$ is full rank, i.e.
\begin{eqnarray}
\label{eqnrank}
    \text{Rank}({\bf A}) & = & P \,,\, \text{if} \,\, (\ref{eqnntp}) \,\, \text{and} \,\, (\ref{eqndistinct}).
\end{eqnarray}Since the channel path delay and Doppler shifts depend on the spatial geometry of reflectors and is not related to the bandwidth $B'$ and time duration $T'$, we assume that the difference between any two path delay shifts is not an integer multiple of $1/B'$ and also that the difference between any two Doppler shifts is not an integer multiple of $1/T'$, i.e. for all $i \ne i'$, $i, i' = 1,2,\cdots, P$
\begin{eqnarray}
    B' (\tau_i - \tau_{i'}) & \notin & {\mathbb Z} \,,\, T' (\nu_i - \nu_{i'}) \, \notin \, {\mathbb Z}.
\end{eqnarray}This then implies that for all $i \ne i'$, $i, i' = 1,2,\cdots, P$
\begin{eqnarray}
    \label{eqndistinctangles}
 \alpha_i & \ne & \alpha_{i'} \,,\, \beta_{i} \ne \beta_{i'}.   
\end{eqnarray}
We summarize the assumptions in the following
\begin{itemize}
\item {\textbf{Assumption} ${\bf 1}$}: We consider prediction over a time span and frequency span over which there is almost no variation in the DD spreading function $h_{\text{phy}}(\tau, \nu)$ of the physical channel. For a practical discussion please refer to Section \ref{subsecstationarity}.
\item \textbf{Assumption} ${\bf 2}$: We assume that $B$ and $T$ are sufficiently high and the channel path delay and Doppler shifts are fractional, so that $N_t \geq P$ (see (\ref{eqnntp})).
\item \textbf{Assumption} ${\bf 3}$: The channel paths are distinct so that ${\bf A}$ has rank $P$ (see (\ref{eqnrank})).
\item \textbf{Assumption} ${\bf 4}$: No two path delay shifts are equivalent modulo $1/B'$ and no path Doppler shifts are equivalent modulo $1/T'$, i.e., the parameters $\alpha_i$, $i=1,2,\cdots, P$ are all distinct and similarly $\beta_i, i=1,2,\cdots, P$ are all distinct.
\end{itemize}

\subsection{Proposed prediction method}
\label{subsecpredict}
Consider the matrix ${\bf H}_{t,q_1, q_2} \in {\mathbb C}^{N_t \times (q_1 - q_2 + 1)}$ given by (\ref{eqn92774}) (see top of next page), with $q_1 \geq q_2 \geq 0$. The columns of this matrix are simply the effective DD domain channel filter in the $(n,0)$-th frame, $n=-q_1, \cdots, -q_2$, i.e., frames in the past but same frequency as the $(0,0)$-th frame. The following result gives a useful factorization of ${\bf H}_{t,q_1, q_2}$.
\begin{figure*}
\vspace{-8mm}
\begin{eqnarray}
\label{eqn92774}
       {\bf H}_{t,q_1, q_2} & \Define \begin{bmatrix}
        h_{-q_1, 0}[k_1,l_1] & \dots & h_{-(q_2+1), 0}[k_1,l_1] & h_{-q_2, 0}[k_1,l_1] \\
        h_{-q_1,0}[k_2,l_2] & \dots & h_{ -(q_2 + 1),0}[k_2, l_2]  &
        h_{-q_2,0}[k_2, l_2] \\
        \vdots & \vdots & \vdots & \vdots \\
        h_{-q_1,0}[k_{N_t},l_{N_t}] & \dots & h_{-(q_2 + 1), 0}[k_{N_t},l_{N_t)}] & h_{-q_2, 0}[k_{N_t},l_{N_t}] \\
    \end{bmatrix}
\end{eqnarray}
\vspace{-4mm}
\begin{eqnarray*}
\hline
\end{eqnarray*}
\end{figure*}
\begin{lemma}
\label{lem1}
For any $q_1 \geq q_2 \geq 0$, the matrix ${\bf H}_{t,q_1, q_2}$ of effective DD domain channel filters for the past frames $(-q_1,0), \cdots, (-q_2,0)$ can be factorized as 
    \begin{eqnarray}
\label{eqn38}
    {\bf H}_{t,q_1, q_2} & = & {\bf A} \, {\bf \Phi}_{q_1, q_2},
    \end{eqnarray}where the time-phase matrix ${\bf \Phi}_{q_1, q_2}$ is of Vandermonde-type and is given by
  \begin{eqnarray}
\label{eqn37}
    {\bf \Phi}_{q_1, q_2} & \Define 
    \begin{bmatrix}
         \alpha_1^{-q_1} &  \dots &  \alpha_1^{-(q_2+1)} &  \alpha_1^{- q_2} \\
          \alpha_2^{-q_1} &  \dots &  \alpha_2^{-(q_2+1)} &  \alpha_2^{- q_2}\\
        \vdots & \vdots & \vdots \\
         \alpha_P^{-q_1} &  \dots &  \alpha_P^{-(q_2+1)} &  \alpha_P^{- q_2} \\
    \end{bmatrix}.
\end{eqnarray}  Under assumptions $3$ and $4$ and with $N_t \geq (q_1 - q_2 + 1) \geq P$
\begin{eqnarray}
\label{rankcondition}
    \textbf{Rank}({\bf H}_{t,q_1, q_2}) & = & P, \,\, \nonumber \\
    & & \hspace{-30mm} \text{given} \,\, \text{Assump.} \, (3), (4) \, \text{and} \, N_t \geq (q_1 - q_2 + 1) \geq P.
\end{eqnarray}
However, if $N_t \geq P > (q_1 - q_2+1)$, then the rank of ${\bf H}_{t,q_1,q_2}$ is at most $(q_1 - q_2 + 1)$ which is less than $P$, i.e.
\begin{eqnarray}
\label{rankcondition2}
    \textbf{Rank}({\bf H}_{t,q_1, q_2}) & < & P, \,\, \nonumber \\
    & & \hspace{-32mm} \text{if} \,\, \text{Assump.} \, (3), (4) \, \text{and} \, N_t \geq P > (q_1 - q_2 + 1).
\end{eqnarray}
\end{lemma}
\begin{IEEEproof}
See Appendix \ref{appen_lem1}.
\end{IEEEproof}
Similarly, consider the matrix ${\bf H}_{f,q_1, q_2} \in {\mathbb C}^{N_t \times (q_1 - q_2 + 1)}$ given by (\ref{eqn92779}) (see top of this page), with $q_1 \geq q_2 \geq 0$. The columns of this matrix are simply the effective DD domain channel filter in the $(0,n)$-th frame, $n=-q_1, \cdots, -q_2$, i.e., frames in the same time duration
as the $(0,0)$-th frame, but different frequency. The following result gives a useful factorization of ${\bf H}_{f,q_1, q_2}$.
\begin{figure*}
\vspace{-8mm}
\begin{eqnarray}
\label{eqn92779}
       {\bf H}_{f,q_1, q_2} & \Define \begin{bmatrix}
        h_{0,-q_1}[k_1,l_1] & \dots & h_{0,-(q_2+1)}[k_1,l_1] & h_{0,-q_2}[k_1,l_1] \\
        h_{0,-q_1}[k_2,l_2] & \dots & h_{0, -(q_2 + 1)}[k_2, l_2]  &
        h_{0,-q_2}[k_2, l_2] \\
        \vdots & \vdots & \vdots & \vdots \\
        h_{0,-q_1}[k_{N_t},l_{N_t}] & \dots & h_{0,-(q_2 + 1)}[k_{N_t},l_{N_t)}] & h_{0,-q_2}[k_{N_t},l_{N_t}] \\
    \end{bmatrix}
\end{eqnarray}
\vspace{-4mm}
\begin{eqnarray*}
\hline
\end{eqnarray*}
\end{figure*}
\begin{lemma}
\label{lem2}
For any $q_1 \geq q_2 \geq 0$, the matrix ${\bf H}_{f,q_1, q_2}$ of effective DD domain channel filters for frames $(0,-q_1), \cdots, (0,-q_2)$ can be factorized as
    \begin{eqnarray}
\label{eqn383975}
    {\bf H}_{f,q_1, q_2} & = & {\bf A} \, {\bf \Theta}_{q_1, q_2},
    \end{eqnarray}where the frequency-phase matrix ${\bf \Theta_{q_1, q_2}}$ is given by
  \begin{eqnarray}
\label{eqn37308}
    {\bf \Theta}_{q_1, q_2} & \Define 
    \begin{bmatrix}
         \beta_1^{q_1} &  \dots &  \beta_1^{(q_2+1)} &  \beta_1^{q_2} \\
          \beta_2^{q_1} &  \dots &  \beta_2^{(q_2+1)} &  \beta_2^{ q_2}\\
        \vdots & \vdots & \vdots \\
         \beta_P^{q_1} &  \dots &  \beta_P^{(q_2+1)} &  \beta_P^{q_2} \\
    \end{bmatrix}.
\end{eqnarray}Under assumptions $3$ and $4$ and with $N_t \geq (q_1 - q_2 + 1) \geq P$
\begin{eqnarray}
\label{rankcondition}
    \textbf{Rank}({\bf H}_{f,q_1, q_2}) & = & P, \,\, \nonumber \\
    & & \hspace{-30mm} \text{given} \,\, \text{Assump.} \, (3), (4) \, \text{and} \, N_t \geq (q_1 - q_2 + 1) \geq P.
\end{eqnarray}
\end{lemma}
\begin{IEEEproof}
Proof is similar to that of Lemma \ref{lem1}.
\end{IEEEproof}
The next result shows that the time-phase matrices and the frequency-phase matrices satisfy rotational invariance.
\begin{lemma}
    \label{lem3}
    For any $Q > 0$, the matrix ${\bf H}_{t,Q-1,0}$ represents the channel time-shifted by one frame relative to ${\bf H}_{t,Q,1}$, due to which the corresponding time-phase matrices
    ${\bf \Phi}_{(Q-1), 0}$ and ${\bf \Phi}_{Q,1}$ satisfy rotational invariance
    \begin{eqnarray}
        {\bf \Phi}_{(Q-1), 0} & = & {\bf \Lambda}_{\alpha} \, {\bf \Phi}_{Q,1}, \nonumber \\
        {\bf \Lambda}_{\alpha} & \Define & \begin{bmatrix}
         \alpha_1 & 0 & \dots & 0 \\
         0 & \alpha_2 & \dots & 0 \\
         \vdots & \vdots & \vdots & \vdots \\
         0 & 0 & \dots & \alpha_P
     \end{bmatrix}.
    \end{eqnarray}Similarly, the frequency-phase matrices ${\bf \Theta}_{(Q-1), 0}$ and ${\bf \Theta}_{Q,1}$ also satisfy rotational invariance
    \begin{eqnarray}
        {\bf \Theta}_{Q, 1} & = & {\bf \Lambda}_{\beta} \, {\bf \Theta}_{Q-1,0}, \nonumber \\
        {\bf \Lambda}_{\beta} & \Define & \begin{bmatrix}
         \beta_1 & 0 & \dots & 0 \\
         0 & \beta_2 & \dots & 0 \\
         \vdots & \vdots & \vdots & \vdots \\
         0 & 0 & \dots & \beta_P
     \end{bmatrix}.
    \end{eqnarray}
\end{lemma}
\begin{IEEEproof}
Follows directly from (\ref{eqn37}) and (\ref{eqn37308}).
\end{IEEEproof}
In the following we discuss the proposed method under assumptions $1,2,3$ and $4$.
We also assume that $N_t \geq Q \geq P$.
From the above lemmas we have, for any $Q > 0$
\begin{eqnarray}
    {\bf H}_{t,Q,1} = {\bf A} \, {\bf \Phi}_{Q,1} & , & {\bf H}_{t,Q-1,0} = {\bf A} \, {\bf \Phi}_{Q-1,0}, \nonumber \\
    {\bf \Phi}_{(Q-1), 0} & = & {\bf \Lambda}_{\alpha} \, {\bf \Phi}_{Q,1}.
\end{eqnarray}Note that the column space of both ${\bf H}_{t,Q,1}$ and ${\bf H}_{t,Q-1,0}$ is same as the column space of ${\bf A}$, and further that ${\bf \Phi}_{Q-1,0}$ and ${\bf \Phi}_{Q,1}$ satisfy rotational invariance. In the following we propose an ESPIRIT (Estimation of Signal Parameters via Rotational Invariance Techniques \cite{espirit_paper}) type method to estimate $\Lambda_{\alpha}$ and ${\bf A}$ from ${\bf H}_{t,Q,1}$ and ${\bf H}_{t,Q-1,0}$. Note that ${\bf H}_{t,Q,1}$ and ${\bf H}_{t,Q-1,0}$ are acquired from the effective DD domain channel filter estimated in the training frames ${\mathcal F}$ (acquisition of channel estimates in the
training frames is briefly reviewed in Section \ref{subsectrainest}).

From Lemma \ref{lem1}, \ref{lem2} and \ref{lem3} it follows that
\begin{eqnarray}
\label{eqn4444}
    {\bf H}_t & \Define & \begin{bmatrix}
        {\bf H}_{t,Q,1} \\
        {\bf H}_{t,Q-1,0} \\
    \end{bmatrix} \, = \, \begin{bmatrix}
         {\bf A} {\bf \Phi}_{Q,1} \\
         {\bf A} {\bf \Lambda}_{\alpha} {\bf \Phi}_{Q,1} \\
    \end{bmatrix}, \nonumber \\
    & = & \begin{bmatrix}
         {\bf A}  \\
         {\bf A} {\bf \Lambda}_{\alpha} \\
    \end{bmatrix} \,  {\bf \Phi}_{Q,1}.
\end{eqnarray}The singular value decomposition (SVD) of ${\bf H}_t \in {\mathcal C}^{2N_t \times Q}$ is given by
\begin{eqnarray}
\label{eqn4545}
    {\bf H}_t & = & {\bf U}_t  {\bf S}_t {\bf V}_t^H,
\end{eqnarray}where the first $P$ columns of ${\bf U}_t$ spans the column space of ${\bf H}_t$. From
    (\ref{eqn4444}) it follows that the matrix $\begin{bmatrix}
         {\bf A}  \\
         {\bf A} {\bf \Lambda}_{\alpha} \\
    \end{bmatrix}$ has $P$ columns with full column rank $P$ since ${\bf A} \in {\mathbb C}^{N_t \times P}$ has full column rank $P$ and ${\bf \Lambda}_{\alpha}$ is a full rank $P \times P$ diagonal matrix.
    Further from (\ref{eqn4444}) and (\ref{eqn4545}) it follows that the column space of $\begin{bmatrix}
         {\bf A}  \\
         {\bf A} {\bf \Lambda}_{\alpha} \\
    \end{bmatrix}$ is same as the column space of the first $P$ columns of ${\bf U}_t$ denoted by 
    \begin{eqnarray}
        {\bf U}_{t,s} & = & \begin{bmatrix}
            {\bf U}_{t,1} \\
            {\bf U}_{t,2} \\
        \end{bmatrix}, \nonumber \\
        {\bf U}_{t,1} & \Define & {\bf U}_t(1:N_t,1:P), \nonumber \\
         {\bf U}_{t,2} & \Define & {\bf U}_t(N_t+1:2N_t,1:P),
    \end{eqnarray}i.e.
    \begin{eqnarray}
        \text{col. space}\left( {\bf U}_{t,s} \right) & = & \text{col. space}\left( \begin{bmatrix}
         {\bf A}  \\
         {\bf A} {\bf \Lambda}_{\alpha} \\
    \end{bmatrix}\right).
    \end{eqnarray}Therefore, there exists a unique invertible transformation matrix ${\bf T}_t \in {\mathbb C}^{P \times P}$ such that
    \begin{eqnarray}
        {\bf U}_{t,s} & = & \begin{bmatrix}
         {\bf A}  \\
         {\bf A} {\bf \Lambda}_{\alpha} \\
    \end{bmatrix} \, {\bf T}_t,
    \end{eqnarray}i.e.
    \begin{eqnarray}
        {\bf U}_{t,1} & = & {\bf A} {\bf T}_t \,,\, \nonumber \\
        {\bf U}_{t,2} & = &  {\bf A} {\bf \Lambda}_{\alpha} \, {\bf T}_t \, = \, {\bf A} {\bf T}_t \, {\bf T}_t^{-1} {\bf \Lambda}_{\alpha} \, {\bf T}_t.
    \end{eqnarray}From this it follows that
    \begin{eqnarray}
        {\bf U}_{t,2} & = & {\bf U}_{t,1} \, {\bf \Psi}_t\,\, , \,\, {\bf \Psi}_t \, \Define \, {\bf T}_t^{-1} {\bf \Lambda}_{\alpha} \, {\bf T}_t.
    \end{eqnarray}This then implies that
    \begin{eqnarray}
      {\bf \Psi}_t & = & {\bf U}_{t,1}^\dagger \, {\bf U}_{t,2},
    \end{eqnarray}where
    \begin{eqnarray}
        {\bf U}_{t,1}^\dagger & \Define & \left({\bf U}_{t,1}^H {\bf U}_{t,1}\right)^{-1} {\bf U}_{t,1}^H
    \end{eqnarray}is the Moore-Penrose pseudo-inverse of ${\bf U}_{t,1}$. Since ${\bf \Psi}_t = {\bf T}_t^{-1} {\bf \Lambda}_{\alpha} {\bf T}_t$, we have
    \begin{eqnarray}
        {\bf \Psi}_t {\bf T}_t^{-1} & = & {\bf T_t}^{-1} {\bf \Lambda}_{\alpha}
    \end{eqnarray}
    which implies that the diagonal elements of the diagonal matrix ${\bf \Lambda}_{\alpha}$, i.e., $\alpha_i, i=1,2,\cdots, P$ are the $P$ eigenvalues of ${\bf \Psi}_t = {\bf U}_{t,1}^{\dagger} {\bf U}_{t,2}$. We use this fact to derive estimates of unit modulus $\alpha_i,i=1,2,\cdots, P$. Further, from (\ref{eqn38}) we know that ${\bf H}_{t,Q-1, 0}  =  {\bf A} \, {\bf \Phi}_{Q-1, 0}$. Since ${\bf \Phi}_{Q-1, 0}$ has full row rank $P$,
    it follows that
    \begin{eqnarray}
    \label{eqn545478}
        {\bf A} & = & {\bf H}_{t,Q-1,0} \, {\bf \Phi}_{Q-1, 0}^\dagger.
    \end{eqnarray}This can be used to estimate ${\bf A}$.

    In a similar manner, the frequency-phase parameters $\beta_i,i=1,2,\cdots,P$ can be estimated from the channel estimates from the $Q$ training frames in frequency, i.e., ${\bf H}_{f, Q,1}$ and ${\bf H}_{f,Q-1, 0}$. To be precise,
    from Lemmas \ref{lem1}, \ref{lem2} and \ref{lem3} it also follows that
    \begin{eqnarray}
    {\bf H}_{f,Q-1,0} = {\bf A} \, {\bf \Theta}_{Q-1,0} & , & {\bf H}_{f,Q,1} = {\bf A} \, {\bf \Theta}_{Q,1}, \nonumber \\
    {\bf \Theta}_{Q, 1} & = & {\bf \Lambda}_{\beta} \, {\bf \Theta}_{Q-1,0}.
\end{eqnarray}Just as the training frames in time can be used to estimate $\alpha_i$, the
estimates of $\beta_i$ are given by the eigenvalues of ${\bf U}_{f,1}^\dagger {\bf U}_{f,2}$, where $\begin{bmatrix}
            {\bf U}_{f,1} \\
            {\bf U}_{f,2} \\
        \end{bmatrix}$ are the first $P$ columns of the left singular matrix ${\bf U}_f$ in the SVD
        \begin{eqnarray}
          \begin{bmatrix}
            {\bf H}_{f,Q-1,0} \\
            {\bf H}_{f,Q,1} \\
        \end{bmatrix}  & = & {\bf U}_f {\bf S}_f {\bf V}_f^H.
        \end{eqnarray}
        Further, from (\ref{eqn383975}) we can estimate ${\bf A}$ using
    \begin{eqnarray}
    \label{eqn54549}
        {\bf A} & = & {\bf H}_{f,Q-1,0} \, {\bf \Theta}_{Q-1, 0}^\dagger.
    \end{eqnarray}
    
Based on the above, in the following we list the steps of the proposed prediction method.
\subsubsection{Step-$1$: Estimation of the number of dominant channel paths}
From Lemma \ref{lem1}, we know that the rank of ${\bf H}_{t,q-1, 0} = P$ only if $q \geq P$. For $q < P$, rank of ${\bf H}_{t,q-1, 0}$ is strictly less than $P$. This gives us a method to estimate $P$ as follows. Starting with $q=1$, we assemble the estimated channel filters from the training frames into the matrix $\widehat{{\bf H}}_{t,q-1, 0}$ which is same as (\ref{eqn92774}) but with its elements replaced by the estimated channel filters
\begin{eqnarray}
\label{step1eqn}
    \widehat{h}_{n,m}[k,l], (n,m) \in {\mathcal F} \,,\, (k,l) \in {\mathcal S}.
\end{eqnarray}Estimation of the effective channel filters in the training frames briefly reviewed in Section \ref{subsectrainest}.
 Through QR factorization of $\widehat{{\bf H}}_{t,q-1, 0}$ we find the number of its dominant linearly independent columns (equal to the number of dominant diagonal entries of the upper triangular matrix in QR factorization). We do this for increasing $Q$ starting with $Q=1$ and stop when there is no further increase in the number of linearly independent columns. This maximal number of dominant linearly independent columns is the estimate of the number of dominant channel paths, denoted by $\widehat{P}$. The QR factorization complexity of $\widehat{{\bf H}}_{t,q-1, 0} \in {\mathbb C}^{N_t \times q}$ is $O(N_tq^2)$ \cite{Golub, Horn13}.
  Therefore the overall complexity of estimating the number of dominant paths in the physical channel
is $O(N_t Q^3)$.
For typical operating values of the pilot power to noise ratio in the training frames, it is observed that $\widehat{P} = P$ with high probability.
\subsubsection{Step-$2$: Estimation of the parameters $\alpha_i, i=1,2,\cdots, \widehat{P}$}
We choose $Q$ such that $N_t > Q > \widehat{P}$.
Let $\widehat{\bf H}_t \, \Define \, \begin{bmatrix}
        \widehat{\bf H}_{t,Q,1} \\
        \widehat{\bf H}_{t,Q-1,0} \\
    \end{bmatrix}$.
Let $\widehat{\bf U}_{t,s} \begin{bmatrix}
    \widehat{\bf U}_{t,1} \\
    \widehat{\bf U}_{t,2}
\end{bmatrix}$ denote the first $\widehat{P}$ columns of the matrix $\widehat{{\bf U}}_t$ of left singular vectors in the SVD
$\widehat{\bf H}_t = \widehat{\bf U}_{t} \widehat{\bf S}_t \widehat{\bf V}_t^H$.
Further, let
\begin{eqnarray}
\label{eigalpha}
    (e_{t,1}, \cdots, e_{t,\widehat{P}}) & \Define & \text{Eigenvalues of} \, \widehat{\bf U}_{t,1}^\dagger \, {\widehat{\bf U}_{t,2}}.
\end{eqnarray}Due to noisy channel estimation in the training frames, the eigenvalues almost never have unit modulus and so we set the estimated $\widehat{\alpha}_i = e_{t,i}/ \vert e_{t,i} \vert$.
\begin{eqnarray}
\label{alphaest}
    \widehat{\alpha}_i & \Define & \frac{e_{t,i}}{\vert e_{t,i} \vert}, \, i=1,2,\cdots, \widehat{P}.
\end{eqnarray}The overall complexity of step-$2$ is $O(N_t Q^2)$.

\subsubsection{Step-$3$: Estimation of the parameters $\beta_i, i=1,2,\cdots, \widehat{P}$.}
Let $\widehat{\bf H}_f \, \Define \, \begin{bmatrix}
        \widehat{\bf H}_{f,Q-1,0} \\
        \widehat{\bf H}_{f,Q,1} \\
    \end{bmatrix}$.
Let $\widehat{\bf U}_{f,s} \begin{bmatrix}
    \widehat{\bf U}_{f,1} \\
    \widehat{\bf U}_{f,2}
\end{bmatrix}$ denote the first $\widehat{P}$ columns of the matrix $\widehat{{\bf U}}_f$ of left singular vectors in the SVD
$\widehat{\bf H}_f = \widehat{\bf U}_{f} \widehat{\bf S}_f \widehat{\bf V}_f^H$.
Further, let
\begin{eqnarray}
\label{eigbeta}
    (e_{f,1}, \cdots, e_{f,\widehat{P}}) & \Define & \text{Eigenvalues of} \, \widehat{\bf U}_{f,1}^\dagger \, {\widehat{\bf U}_{f,2}}.
\end{eqnarray}Due to noisy channel estimation in the training frames, the eigenvalues almost never have unit modulus and so we set the estimated $\widehat{\widehat{\beta}}_i = e_{f,i}/ \vert e_{f,i} \vert$.
\begin{eqnarray}
    \widehat{\widehat{\beta}}_i & \Define & \frac{e_{f,i}}{\vert e_{f,i} \vert}, \, i=1,2,\cdots, \widehat{P}.
\end{eqnarray}The overall complexity of step-$3$ is $O(N_t Q^2)$.

\subsubsection{Step-$4$: Estimation of $A_i[k,l]$, $i=1,2,\cdots, \widehat{P}$, $(k,l) \in {\mathcal S}$}
Using (\ref{eqn545478}) an estimate of ${\bf A}$ based on the training frames in time is given by
    \begin{eqnarray}
    \label{eqn63}
        \widehat{\bf A}_t & \Define & \widehat{\bf H}_{t,Q-1,0} \, \widehat{\bf \Phi}_{Q-1, 0}^\dagger,
    \end{eqnarray}where
  \begin{eqnarray}
\label{eqn37est}
    \widehat{\Phi}_{Q-1, 0} & \Define 
    \begin{bmatrix}
         \widehat{\alpha}_1^{-(Q-1)} &  \dots &  \widehat{\alpha}_1^{-1} &  1 \\
          \widehat{\alpha}_2^{-(Q-1)} &  \dots &  \widehat{\alpha}_2^{-1} &  1\\
        \vdots & \vdots & \vdots \\
         \widehat{\alpha}_{\widehat{P}}^{-(Q-1)} &  \dots &  \widehat{\alpha}_{\widehat{P}}^{-1} &  1 \\
    \end{bmatrix}.
\end{eqnarray}

In a similar manner, based on (\ref{eqn54549}) we derive another estimate of ${\bf A}$ based on the frequency domain training frames and the estimated parameters $\widehat{\widehat{\beta}}_i, i=1,2,\cdots, \widehat{P}$, which is given by
\begin{eqnarray}
    \label{eqnaf}
    \widehat{{\bf A}}_f & \Define & {\bf H}_{f,Q-1, 0} \, \widehat{\bf \Theta}_{Q-1, 0}^\dagger,  
\end{eqnarray}where
\begin{eqnarray}
\label{eqn39est}
    \widehat{\bf \Theta}_{Q-1, 0} & \Define 
    \begin{bmatrix}
         \widehat{\widehat{\beta}}_1^{-(Q-1)} &  \dots &  \widehat{\widehat{\beta}}_1^{-1} &  1 \\
          \widehat{\widehat{\beta}}_2^{-(Q-1)} &  \dots &  \widehat{\widehat{\beta}}_2^{-1} &  1\\
        \vdots & \vdots & \vdots \\
         \widehat{\widehat{\beta}}_{\widehat{P}}^{-(Q-1)} &  \dots &  \widehat{\widehat{\beta}}_{\widehat{P}}^{-1} &  1 \\
    \end{bmatrix}.
\end{eqnarray}Even in the absence of noise, i.e. with perfect channel estimates in the training frames, the columns in $\widehat{{\bf A}}_t$ and $\widehat{{\bf A}}_f$ although being same may be ordered differently, i.e., the first column in $\widehat{{\bf A}}_t$
may not be the first column in
$\widehat{{\bf A}}_f$. This is because in (\ref{eigalpha}) and (\ref{eigbeta}), the eigenvalues $e_{t,i},i=1,2,\cdots,\widehat{P}$ and $e_{f,i},i=1,2,\cdots,\widehat{P}$ could be permuted with respect to each other. To resolve the ordering ambiguity between $\widehat{{\bf A}}_t$ and $\widehat{{\bf A}}_f$, we need to find a $\widehat{P} \times \widehat{P}$ permutation matrix $\Tilde{\bf P}$ such that
\begin{eqnarray}
    \widehat{{\bf A}}_t & \approx & \widehat{{\bf A}}_f \, \Tilde{\bf P}.
\end{eqnarray}
We use the Hungarian algorithm to find the optimal permutation matrix $\Tilde{\bf P}$ which is given by
\begin{eqnarray}
    \label{optpermute}
    \Tilde{\bf P} & = & \arg \min_{{\bf G} \in {\mathcal P}} \left\Vert \widehat{{\bf A}}_t - {\bf G} \widehat{{\bf A}}_f \right\Vert_F^2,
\end{eqnarray}where ${\mathcal P}$ is the set of all $\widehat{P}! = \widehat{P} (\widehat{P}-1) \cdots 2. 1$ $\widehat{P} \times \widehat{P}$ permutation matrices. Also, for any matrix ${\bf B}$, $\Vert B \Vert_F$ denotes its Frobenius norm. The Hungarian algorithm
given the optimal solution for the job assignment problem \cite{E72}, where $\widehat{P}$ jobs are to be assigned to $\widehat{P}$ workers with exactly one job assigned to each worker. The cost of getting the $i$-th job executed by the $j$-th worker is denoted by $C_{i,j}$, $i,j = 1,2,\cdots, \widehat{P}$. Each assignment is a $\widehat{P}$ tuple $(a_1, a_2, \cdots, a_{\widehat{P}})$ with $a_i \in \{ 1,2,\cdots, \widehat{P} \}$, $i=1,2,\cdots, \widehat{P}$, which means that the $i$-th job is assigned to the $a_i$-th worker. The Hungarian algorithm finds the assignment which gives the minimum total cost $\sum\limits_{i=1}^{\widehat{P}} C_{i,a_i}$. 

Our problem of finding the optimal permutation matrix is similar to a job assignment problems where the ${\widehat P}$ columns of $\widehat{{\bf A}}_t$ are the ``jobs" and the
${\widehat P}$ columns of $\widehat{{\bf A}}_f$
are the workers. The cost of assigning/matching the $i$-th column of $\widehat{{\bf A}}_t$ to the $j$-th column of $\widehat{{\bf A}}_f$ is simply $C_{i,j} = \Vert \widehat{A}_{t,i} - \widehat{A}_{f,j} \Vert_2^2$, where $\widehat{A}_{t,i}$ and $\widehat{A}_{f,j}$ denote the $i$-th and $j$-th column of $\widehat{{\bf A}}_t$ and $\widehat{{\bf A}}_f$ respectively. Each assignment is equivalent to a $\widehat{P} \times \widehat{P}$ permutation matrix ${\bf G}$ for which the total assignment/matching cost is therefore $\left\Vert \widehat{{\bf A}}_t - {\bf G} \widehat{{\bf A}}_f \right\Vert_F^2$. 

This optimal permutation is applied to reorder the estimates $\widehat{\widehat{\beta}}_i, i = 1,2,\cdots, \widehat{P}$
resulting in the final estimates
\begin{eqnarray}
\label{eqneqn55}
{\Big [} \widehat{\beta}_1 \, \cdots, \widehat{\beta}_{\widehat{P}} {\Big ]} & \Define & {\Big [} \widehat{\widehat{\beta}}_1 \, \cdots, \widehat{\widehat{\beta}}_{\widehat{P}} {\Big ]} \, \Tilde{\bf P}.
\end{eqnarray}The final estimate of ${\bf A}$ is given by
\begin{eqnarray}
\label{eqneqn56}
    \widehat{\bf A} & \Define & \frac{\left( \widehat{{\bf A}}_t \, + \, \widehat{{\bf A}}_f \, \Tilde{\bf P} \right)}{2}.
\end{eqnarray}From (\ref{eqnadef}), it follows that the estimates of $A_i[k,l]$ are simply the elements of $\widehat{\bf A}$, i.e.
\begin{eqnarray}
    \widehat{\bf A} &\hspace{-3.5mm} = & \hspace{-3.5mm} \begin{bmatrix}
        \widehat{A}_1[k_1, l_1] & \widehat{A}_2[k_1, l_1] & \dots & \widehat{A}_P[k_1, l_1] \\
        \widehat{A}_1[k_2, l_2] & \widehat{A}_2[k_2, l_2] & \dots & \widehat{A}_P[k_2, l_2] \\
        \vdots & \vdots & \vdots & \vdots  \\
         \widehat{A}_1[k_{N_t}, l_{N_t}] & \widehat{A}_2[k_{N_t}, l_{N_t}] & \dots & \widehat{A}_P[k_{N_t}, l_{N_t}].
    \end{bmatrix}
\end{eqnarray}The overall complexity of step-$4$ is $O(N_tPQ)$.

\subsubsection{Step-$5$: Generation of the predicted effective DD domain channel filter for the $(n,m)$-th frame}
Based on the estimates acquired in Step-$2$,$3$, $4$ and $5$, the proposed estimate of the effective DD domain channel filter in the $(n,m)$-th frame is given by
\begin{eqnarray}
\label{eqngenpredict}
    \widehat{h}_{n,m}[k,l] & = & \begin{cases}
        \sum\limits_{i=1}^{\widehat{P}} (\widehat{\alpha}_i)^n \, (\widehat{\beta}_i)^m \, \widehat{A}_i[k,l] &, (k,l) \in {\mathcal S}\\
        0 &, \text{otherwise} \\
    \end{cases}.
\end{eqnarray}
The complexity of generating $\widehat{h}_{n,m}[k,l], (k,l) \in {\mathcal S}$ is therefore
$O(PN_t)$.

Since $N_t \geq Q \geq P$, the overall computational complexity of the proposed prediction method is $O(N_tQ^3)$. Based on exhaustive simulation results, it has been observed that channel prediction made by the proposed method are accurate even $60$ frames into the future and $60$ frames apart in frequency. As we need to carry out the prediction only after every $50-60$ frames in time, the computations involved in the proposed method can be spread over several frames in time, thereby relaxing the computational requirement.

\subsection{Estimation of the effective DD domain channel in training frames}
\label{subsectrainest}
Training frames are regular Zak-OTFS frames as shown in Fig.~\ref{fig_singlepilot}, with a single pilot pulsone at DD location $(k_p, l_p)$ depicted as a red dot with the received channel response being spread over DD taps depicted as a pink ellipse. The DD domain taps in the ``green" data region carry information symbols and are separated from the pilot region through guard regions which do not carry information symbols. The pilot and guards regions are an overhead as they do not carry information. An important aspect of the proposed prediction method is that it allows us to avoid this overhead (increase effective spectral efficiency) in the non-training/prediction frames whose effective DD channel filter is predicted by the proposed method, as these non-training frames do not need a point pilot. As shown in Fig.~\ref{fig_singlepilot},
the overhead is $(2k_{max} + 5)N$ out of $MN$ DD carriers/pulsones, where $k_{max} \Define \lceil B \tau_{max} \rceil$.
Estimation of the effective DD channel filter from the pilot response received in the pilot region of the training frames has been described in \cite{zakotfs2,twobytwopaper}.

    \begin{figure}
     \centering
        \includegraphics[width=9.0cm,height=6.0cm]{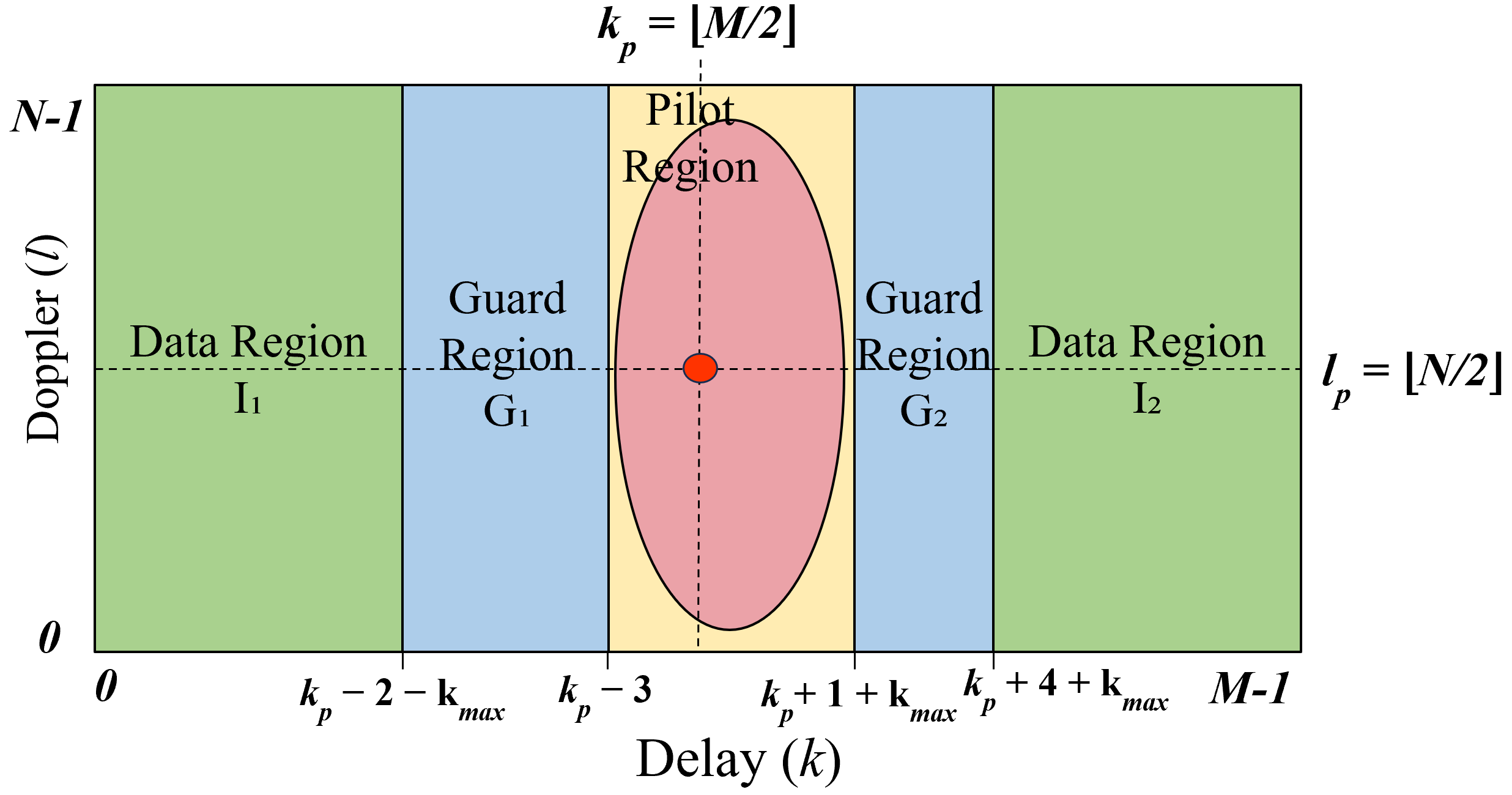}
        \vspace{-2mm}
        \caption{Zak-OTFS frame in DD domain with single pilot (depicted by a red dot) at $(k_p, l_p)$. The pink shaded ellipse depicts the support set of the channel response to the pilot (i.e., ${\mathcal S} \, + \, (k_p, l_p)$). Figure taken from our work in \cite{twobytwopaper}.}
        \label{fig_singlepilot}
        \vspace{-3mm}
    \end{figure}
\subsection{On the stationarity of the DD domain spreading function $h_{\text{phy}}(\tau, \nu)$ of the physical channel}
\label{subsecstationarity}
The proposed method gives accurate prediction
for frames ahead in time only as long as the time duration $nT'$ between the current frame $(0,0)$ and the future frame $(n,m)$ is not so high that the paths shift from one DD tap to the other. This could happen for example if the distance between the transmitter and a mobile receiver increases by $c/B$ so that the delay increases by $1/B$, i.e., the path moves to the next delay tap. If the receiver is moving at a constant velocity $v$ then this can happen in about $c/(vB)$ seconds. As an example with $B = 1$ MHz and $v = 30$ m/s, $c/(vB) = 10$ seconds. Even for $B=10$ MHz this time duration is $1$ second. Another reason could be due to change in the receiver's velocity which could be either due to linear acceleration or due to change in the angle of arrival of the signal at the mobile receiver. Each Doppler tap is $1/T$ Hz. For $T = 1$ ms, this is $1$ kHz, and therefore the Doppler of any path has to shift by $1$ kHz for it to move to the next Doppler bin. Even for a path with a Doppler shift of $2$ kHz, the cosine of the angle between the velocity vector of the receiver and the direction of signal arrival has to change by $0.5$, for example a change in this angle from zero to $60$ degrees. This could happen for example if the mobile receiver moves around a corner at a road intersection which will still need at least one-tenth of a second. Therefore, in this paper we assume that the DD spreading function is almost stationary for fifty to sixty milliseconds, which allows us to predict the effective DD channel roughly up to fifty milliseconds in future. 

Regarding the prediction for frames $(0,m)$ at current time but at frequency $mB'$ Hz away from the $(0,0)$-th frame, we are interested in the stationarity of $h_{\text{phy}}(\tau, \nu)$ across the frequency domain. One possible source of variation in frequency could be due to the non-constant response of transmit and receiver filters which will still be several hundreds of MHz.

\subsection{On the choice of the support set ${\mathcal S}$}
\label{subsecnt}
We form an average energy profile for the DD taps of the effective DD domain channel filter using the DD channel estimates obtained from the set of training frames, i.e.
\begin{eqnarray}
    E[k,l] & \Define & \frac{1}{\vert {\mathcal F} \vert } \sum\limits_{(n,m) \in {\mathcal F}}  \left\vert  \widehat{h}_{n,m}[k,l] \right\vert^2,
\end{eqnarray}for all $(k,l) \in {\mathcal S}'$, where ${\mathcal S}'$ is a large set of DD taps which covers almost all energy in the effective DD domain channel filter, i.e.\footnote{\footnotesize{We assume sufficiently high $T$ and $B$ such that each DD tap receives most energy from only a single or very few paths. Since $\alpha_i$ and $\beta_i$ are unit modulus, we have observed that for any $(k,l) \in {\mathbb S}$, $\vert h_{n,m}[k,l] \vert^2$ changes very slowly with $(n,m)$ and we therefore assume that it is almost stationary for the TF region of the training frames and we predict the effective DD channel filter for only those frames where this stationarity holds.}}
\begin{eqnarray}
    \sum\limits_{(k,l) \in {\mathcal S}'} \vert h_{0,0}[k,l] \vert^2 & \approx & \sum\limits_{(k,l) \in {\mathbb Z}} \vert h_{0,0}[k,l] \vert^2.
\end{eqnarray}We then sort the taps based on their energies $E[k,l]$, such that the list of sorted taps in descending order of energy is given by $(k_1,l_1), (k_2, l_2), \cdots, $. We then choose the number of taps $N_t$ to be the greatest positive integer $u$ such that $E[k_{u},l_u]/E[k_1, l_1]$ is greater than $X$ dB. We choose $X$ based on apriori knowledge of the channel power gain ratio between the strongest and the weakest physical channel path, i.e., ${\mathbb E}[\vert h_P \vert^2]/{\mathbb E}[\vert h_1 \vert^2]$. In this paper we choose $X = 0.01 \, {\mathbb E}[\vert h_P \vert^2]/{\mathbb E}[\vert h_1 \vert^2]$, which is $-40$ dB for the Vehicular-A channel whose power delay profile is shown in Table \ref{tabveha}. ${\mathcal S}$ is then given by $\{ (k_1, l_1), (k_2, l_2), \cdots, (k_{N_t}, l_{N_t}) \}$. Also, another constraint is that we must choose $N_t$ such that it is greater than both $Q$ and $P$ (for the proposed method to work). If this is not the case then we relax $X$ so that this constraint is just satisfied.

\begin{table}[!t]
\vspace{-2mm}
    \centering
    \caption{Power-delay profile of Veh-A channel model}
    \vspace{-2mm}
    \begin{tabular}{|c|c|c|c|c|c|c|}
         \hline
         Path index $i$ & 1 & 2 & 3 & 4 & 5 & 6 \\
         \hline
         Delay $\tau_i (\mu s)$ & 0 & 0.31 & 0.71 & 1.09 & 1.73 & 2.51 \\
         \hline
         Relative power (dB) & 0 & -1 & -9 & -10 & -15 & -20 \\
         \hline
    \end{tabular}
    \label{tabveha}
    \vspace{-3mm}
\end{table}

\section{Numerical results}
\label{secsim}
In this section we present simulations results in support of the proposed prediction method for the
$P=6$ path Vehicular-A channel model whose power-delay profile is given in Table \ref{tabveha} \cite{EVAITU}.
The channel path gains $h_i, i=1,2,\cdots, P$ are modeled as zero mean independent circularly symmetric complex Gaussian random variables with variance ${\mathbb E}[\vert h_i \vert^2]$. The ratio ${\mathbb E}[\vert h_i \vert^2]/{\mathbb E}[\vert h_1 \vert^2]$ is equal to the relative power given in Table \ref{tabveha}. The Doppler shift for the $i$-th channel path $\nu_i = \nu_{max} \cos(\theta_i)$, where $\theta_i, i=1,2,\cdots,P$ are i.i.d. uniformly distributed in the interval $[0 \,,\, 2 \pi)$, and $\nu_{max}$ is the maximum possible Doppler shift of any path. If the relative speed between the transmitter and receiver is $v$ m/s then we consider $\nu_{max} = v f_c/c$, where $f_c$ is the carrier frequency and $c$ is the speed of light.

The other simulation parameters are as follows. We consider $\nu_p = 60$ kHz, $\tau_p = 1/\nu_p = 16.66 \mu s$, $M = 50$, $N = 50$, which implies that $B = M \nu_p = 3$ MHz, $T = N \tau_p = 0.8333$ ms. We choose the pulse-shaping filter to be the Root Raised Cosine (RRC) filter which is given by
\cite{SHDigcomm}

{\vspace{-4mm}
\small
\begin{eqnarray}
\label{rrcpulse_eqn1}
w_{tx}(\tau,\nu) & = &  \sqrt{BT} \, rrc_{_{\beta_{\tau}}}( M \nu_p \tau ) \,  rrc_{_{\beta_{\nu}}}( N \tau_p \nu ), \nonumber \\
rrc_{_{\beta}}(x) &  = &  \frac{\sin(\pi x (1 - \beta)) + 4 \beta x \cos(\pi x (1 + \beta))}{\pi x \left( 1 - (4 \beta x)^2 \right)},
\end{eqnarray}\normalsize}where $\beta_{\tau}$ and $\beta_{\nu}$ are the roll-off factors which are chosen to be $\beta_{\tau} = \beta_{\nu} = 0.2$.
Therefore, the time-duration and bandwidth of each frame is $B' = (1 + \beta_{\tau})B = 3.6$ MHz and $T' = (1 + \beta_{\nu}) T = 1$ ms.
The number of DD filter taps $N_t$ depends on $X$ (the ratio of the average energy of the weakest DD tap to that of the strongest tap) which we choose to be $-40$ dB. $N_t$ also depends on the channel delay and Doppler spread. 

    \begin{figure}
     \centering
\includegraphics[width=9.0cm,height=5.6cm]{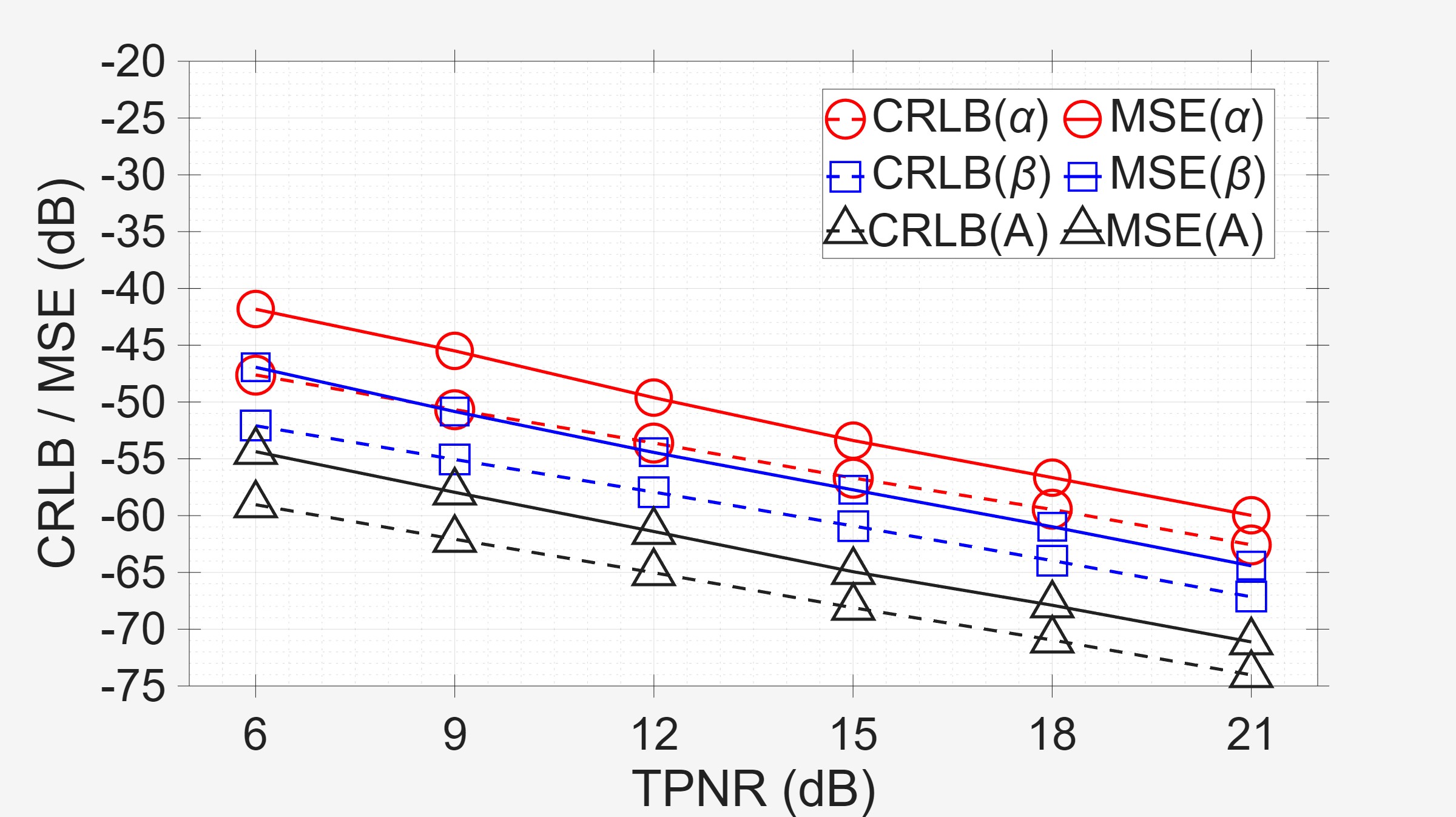}
        \hspace{-3mm}
        \caption{MSE error of the proposed estimation of the time-phase, frequency-phase parameters and the DD signature matrix.}
        \label{figcrlb}
        \vspace{-3mm}
    \end{figure}
In Fig.~\ref{figcrlb} we plot the average Mean Squared Estimation Error (MSE) of the parameters $\alpha_i, \beta_i, i=1,2,\cdots, P$ and $A_i[k,l], (k,l) \in {\mathcal S}$ for the proposed prediction method for $\nu_{max} = 1$ kHz. The average MSE for the time-phase parameters is given by ${\mathbb E}\left[ \sum\limits_{i=1}^P \vert \widehat{\alpha}_i - \alpha_i \vert^2\right]/P$, that for the frequency-phase parameters is given by ${\mathbb E}\left[ \sum\limits_{i=1}^P \vert \widehat{\beta}_i - \beta_i \vert^2\right]/P$, and the MSE for the DD signature matrices is given by ${\mathbb E}\left[ \sum\limits_{i=1}^P \sum\limits_{(k,l) \in {\mathcal S}}\vert \widehat{A}_i[k,l] - A_i[k,l] \vert^2\right]/{\mathbb E}\left[ \sum\limits_{i=1}^P \sum\limits_{(k,l) \in {\mathcal S}}\vert A_i[k,l] \vert^2\right]$. The expectation in the MSE expression is w.r.t.
the random channel realizations, AWGN and data symbols transmitted in the training frames. For a baseline comparison, we also plot the Cramer Rao Lower Bound (CRLB) for these estimates \cite{SM93}. We consider $Q = 30$ (i.e., $2Q+1 = 61$ training frames) and $N_t = 60$ taps (sufficient for $\nu_{max} = 1$ kHz). The pilot power to data power ratio (PDR) in the training frames is $0$ dB.

In Fig.~\ref{figcrlb} we see that the MSE of the proposed ESPIRIT-type estimation method is close to CRLB for a wide range of the total transmit power to noise ratio (TPNR) in the training frames. Since PDR is $0$ dB, the ratio of the pilot power to noise power is $3$ dB less than TPNR. We expect the gap between CRLB and the proposed estimation method to decrease further with increasing $Q$ \cite{Ottersen91}.\footnote{\footnotesize{In the current simulations, the number of parameters to be estimated is $2P + N_t P =372$, and the number of observations (i.e., DD channel taps estimated from the training frames) is $(2Q+1)N_t = 3660$ which is only about ten times larger than the number of parameters. Therefore, the current operating point is far from the asymptote where the number of observations is significantly higher than the number of fixed parameters. Increasing $Q$ can lead us to the asymptote but at the cost of increased prediction complexity and also reduced spectral efficiency since the training frames carry a pilot signal.}} 

    \begin{figure}
     \centering
\includegraphics[width=9.0cm,height=5.6cm]{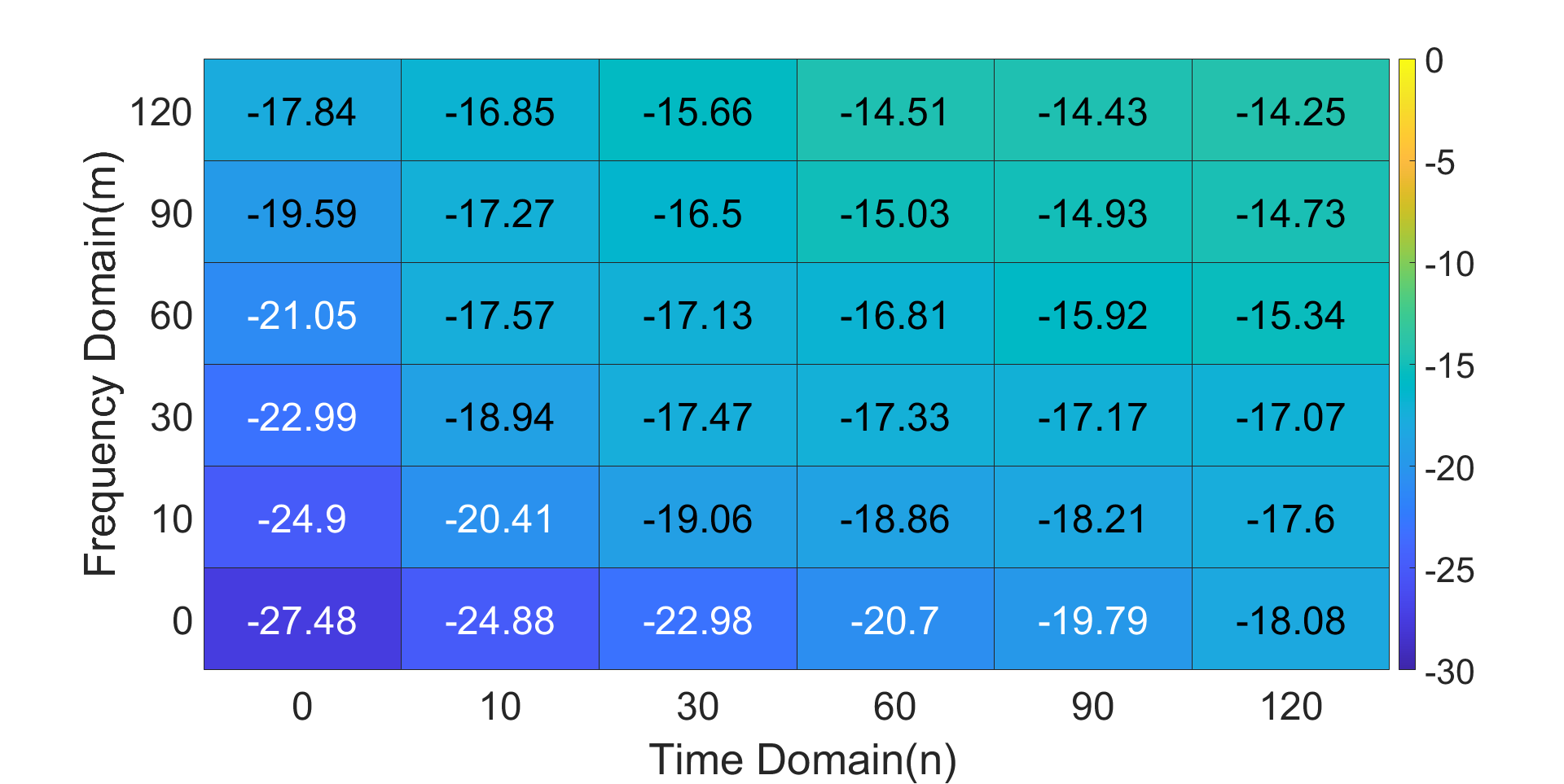}
        \hspace{-3mm}
    \caption{Heatmap of the normalized mean squared prediction error (NMSPE) of the predicted DD domain channel filter $\widehat{h}_{n,m}[k,l]$, as a function of $(n,m)$. TPNR is $18$ dB and PDR is $0$ dB in training frames. Other parameters are same as in Fig.~\ref{figcrlb}.}
        \label{figheatmappredictionhorizon}
        \vspace{-3mm}
    \end{figure}
Next, we study the prediction horizon by measuring the normalized mean squared prediction error (NMSPE) of the predicted DD domain channel filter $\widehat{h}_{n,m}[k,l]$, as a function of $(n,m)$. 
The expression of the predicted $\widehat{h}_{n,m}[k,l]$ is given by (\ref{eqngenpredict}). NMSPE for the $(n,m)$-th prediction frame is given by

{\small
\vspace{-4mm}
\begin{eqnarray}
    \text{NMSPE}(n,m) & \hspace{-3mm} \Define & \hspace{-3mm} \frac{{\mathbb E}\left[ \sum\limits_{k,l \in {\mathbb Z}} \vert \widehat{h}_{n,m}[k,l] - h_{n,m}[k,l] \vert^2 \right]}{{\mathbb E}\left[ \sum\limits_{k,l \in {\mathbb Z}} \vert h_{n,m}[k,l] \vert^2\right]}.
\end{eqnarray}\normalsize}In Fig.~\ref{figheatmappredictionhorizon} we plot the heatmap of NMSPE for increasing $(n,m)$, for a fixed TPNR of $18$ dB and PDR of $0$ dB in the training frames. Other parameters are same as in Fig.~\ref{figcrlb}. As expected, the NMSPE increases with increasing $n$ and $m$, i.e., as we move farther into future in time, or farther away in frequency. The proposed prediction has a long prediction horizon in both time and frequency since the NMSPE for a frame $120 T' = 120$ ms into the future and $120 B' = 432$ MHz away in frequency is still good ($-14.25$ dB).
Even for a low TPNR of $10$ dB (i.e., pilot power to noise power ratio of only $7$ dB in training frames), the NMSPE for the frame at $(n,m) = (120, 120)$ is $-10$ dB (see Fig.~\ref{figheatmappredictionhorizon10dB}). 

The increase in NMSPE with increasing $(n,m)$ is expected since
with increasing $(n,m)$ the difference between the phase of $\widehat{\alpha_i}^n$ and $\alpha_i^n$, and the difference between the phase of 
$\widehat{\beta_i}^m$ and $\beta_i^m$
increases resulting in increasing error between the terms $(\widehat{\alpha}_i)^n \, (\widehat{\beta}_i)^m$ and ${\alpha}_i^n \, \beta_i^m$ in (\ref{eqngenpredict}) and (\ref{eqnthm3}) respectively.
    \begin{figure}
     \centering
\includegraphics[width=9.0cm,height=5.6cm]{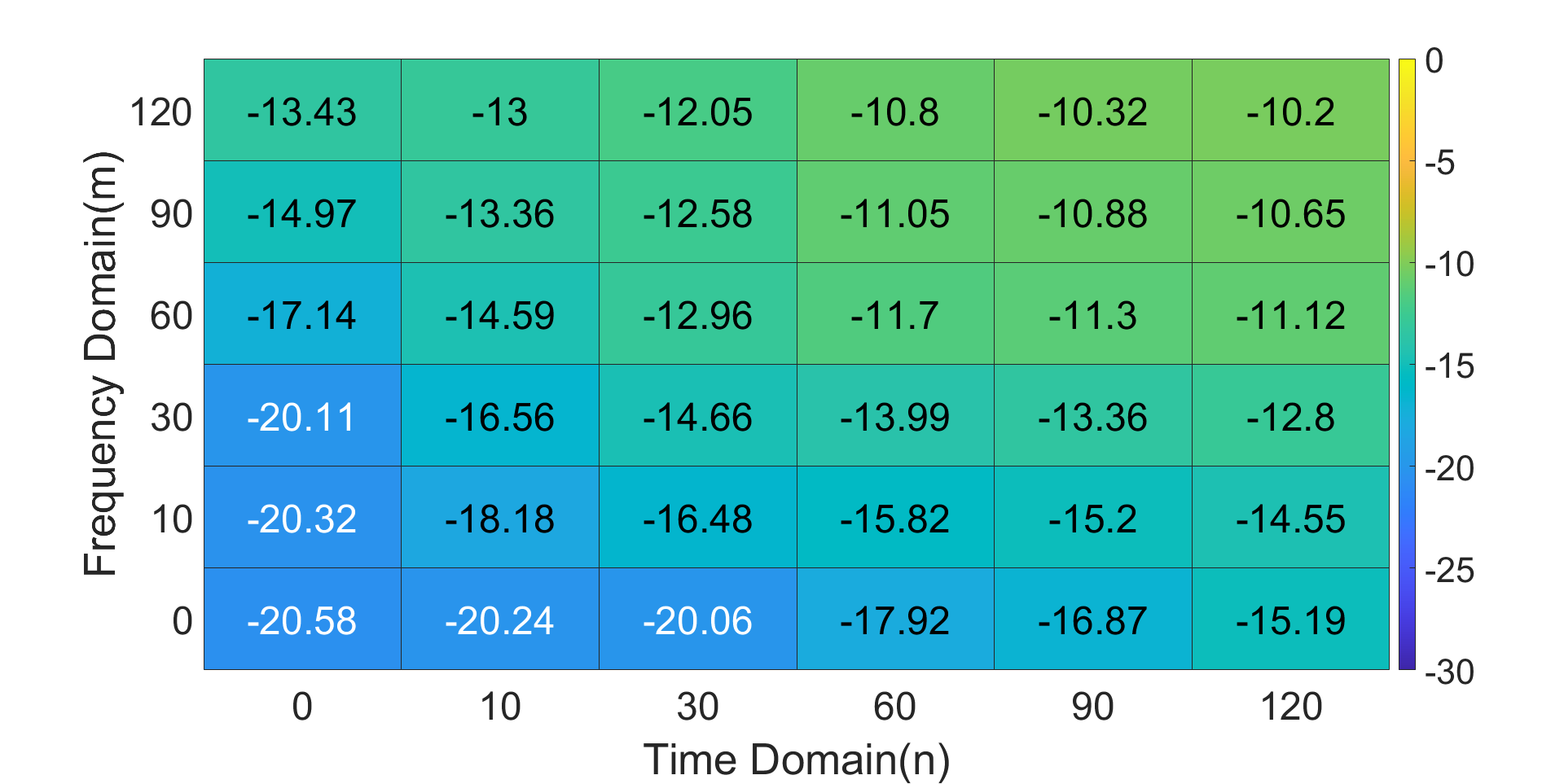}
        \hspace{-3mm}
    \caption{Heatmap of NMSPE of the predicted DD domain channel filter $\widehat{h}_{n,m}[k,l]$, as a function of $(n,m)$. TPNR is $10$ dB and PDR is $0$ dB in training frames. Other parameters are same as in Fig.~\ref{figcrlb}.}
        \label{figheatmappredictionhorizon10dB}
        \vspace{-3mm}
    \end{figure}
    \begin{figure}
     \centering
\includegraphics[width=9.0cm,height=5.6cm]{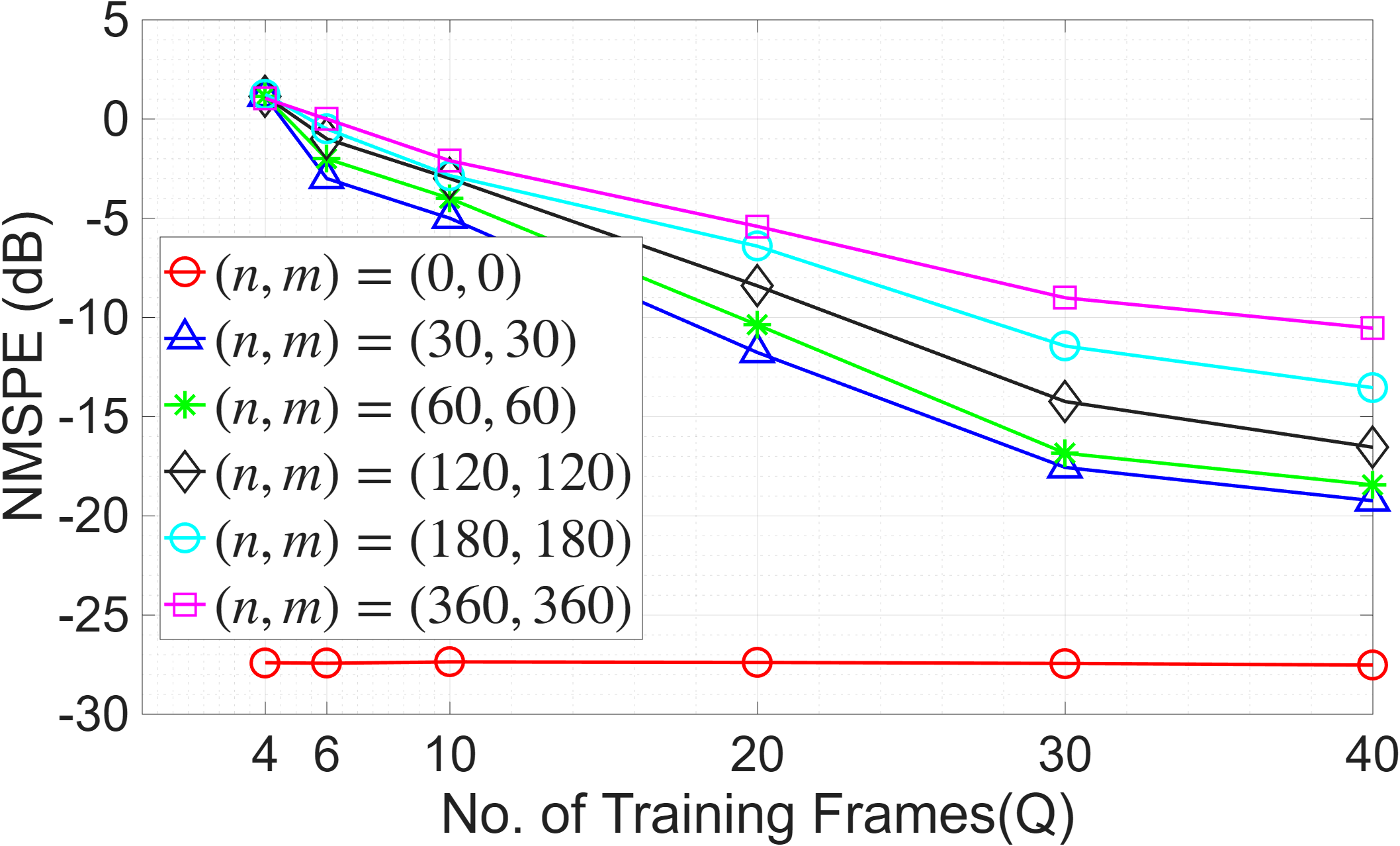}
        \hspace{-3mm}
    \caption{NMSPE vs $Q$. TPNR is $18$ dB and PDR is $0$ dB in training frames. Other parameters are same as in Fig.~\ref{figcrlb}.}
\label{NMSPEQ}
        \vspace{-3mm}
    \end{figure}In Fig.~\ref{NMSPEQ}, we plot the NMSPE as a function of increasing $Q$ for a fixed TPNR of $18$ dB and PDR of $0$ dB in the training frames. As expected, NMSPE decreases steadily with increasing $Q$ since in the estimation of the time-phase, frequency-phase and DD signature matrix parameters, the number of observations of the effective DD domain channel increases with increasing $Q$, while the number of parameters is fixed. We also observe that for a given $Q$, the NMSPE degrades significantly beyond $(n,m) = (120, 120)$. For example, for $Q=30$, the NMSPE for $(n,m) = (360, 360)$ is roughly $-8$ dB when compared to $-14.25$ dB for $(n,m) = (120, 120)$.

    \begin{figure}
     \centering
\includegraphics[width=9.0cm,height=5.6cm]{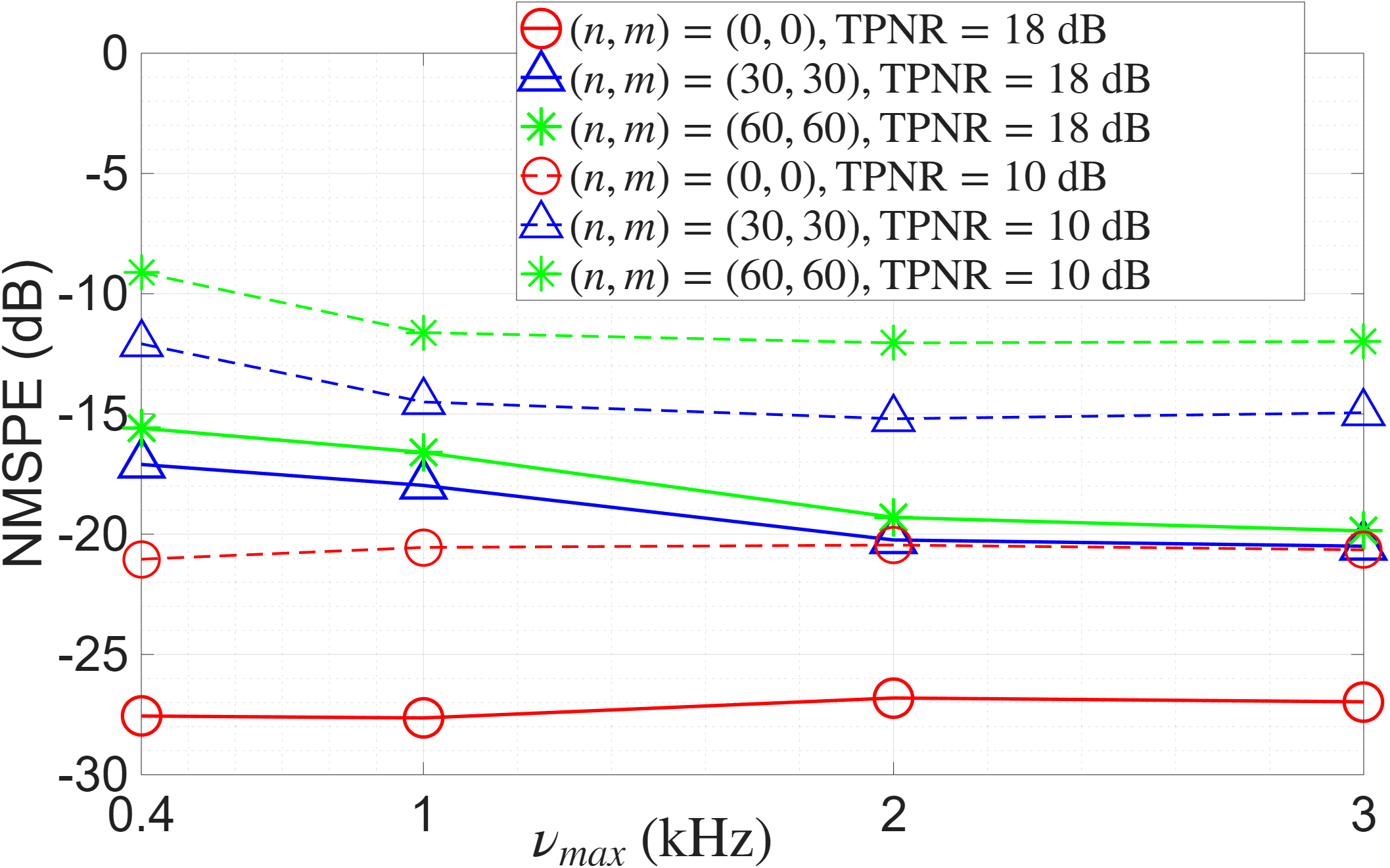}
        \hspace{-3mm}
    \caption{NMSPE vs $\nu_{max}$ (in kHz). $Q=30$. Other parameters are same as in Fig.~\ref{figcrlb}.}
    \vspace{-3mm}
\label{NMSPEnumax}
        \vspace{-3mm}
    \end{figure}In Fig.~\ref{NMSPEnumax}, we plot the NMSPE as a function of increasing $\nu_{max}$ for a fixed TPNR of $18$ dB and $10$ dB and fixed $Q=30$. It is observed that the NMSPE improves with increasing $\nu_{max}$ and is almost constant for high $\nu_{max}$. This is because, with increasing $\nu_{max}$, the path Doppler shifts become increasingly separable, i.e., the probability of $\vert \nu_i - \nu_j \vert$, $i \ne j, i,j =1,2,\cdots, P$ taking values much smaller than $1/T'$ decreases. This implies that the time-phase parameters $\alpha_i, i=1,2,\cdots,P$ are separable/distinct due to which the Vandermonde-type matrix ${\bf \Phi}_{Q-1,0}$ in (\ref{eqn37}) is well-conditioned which results in lower amplification of the estimation error in $\widehat{\bf \Phi}_{Q-1,0}$ when we compute its pseudo-inverse during the estimation of $ \widehat{\bf A}_t  = \widehat{\bf H}_{t,Q-1,0} \, \widehat{\bf \Phi}_{Q-1, 0}^\dagger$ in (\ref{eqn63}).
    \begin{figure}
     \centering
\includegraphics[width=9.0cm,height=5.6cm]{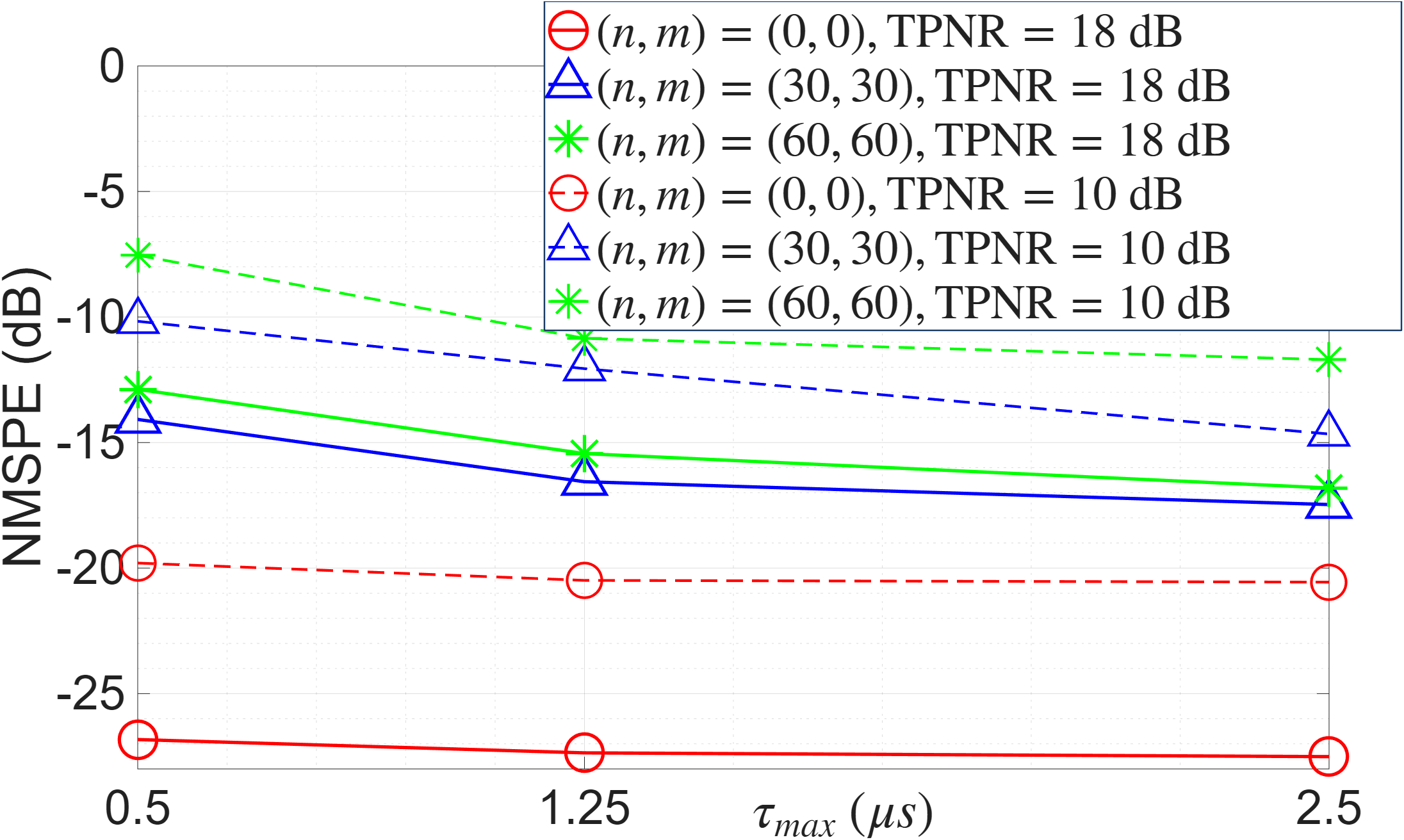}
        \hspace{-3mm}
    \caption{NMSPE vs $\tau_{max}$ (in $\mu s$). Fixed $\nu_{max} = 1$ kHz. Other parameters are same as in Fig.~\ref{figcrlb}. $Q=30$.}
    \vspace{-1mm}
\label{NMSPEtaumax}
        \vspace{-1mm}
    \end{figure}In Fig.~\ref{NMSPEtaumax} we plot NMSPE as a function of increasing maximum channel delay spread $\tau_{max}$ for a fixed $\nu_{max} = 1$ kHz, $Q=30$. It is observed that the NMSPE improves with increasing $\tau_{max}$ since the channel paths becomes more separable along the delay axis.
    \begin{figure}
     \centering
\includegraphics[width=9.0cm,height=5.6cm]{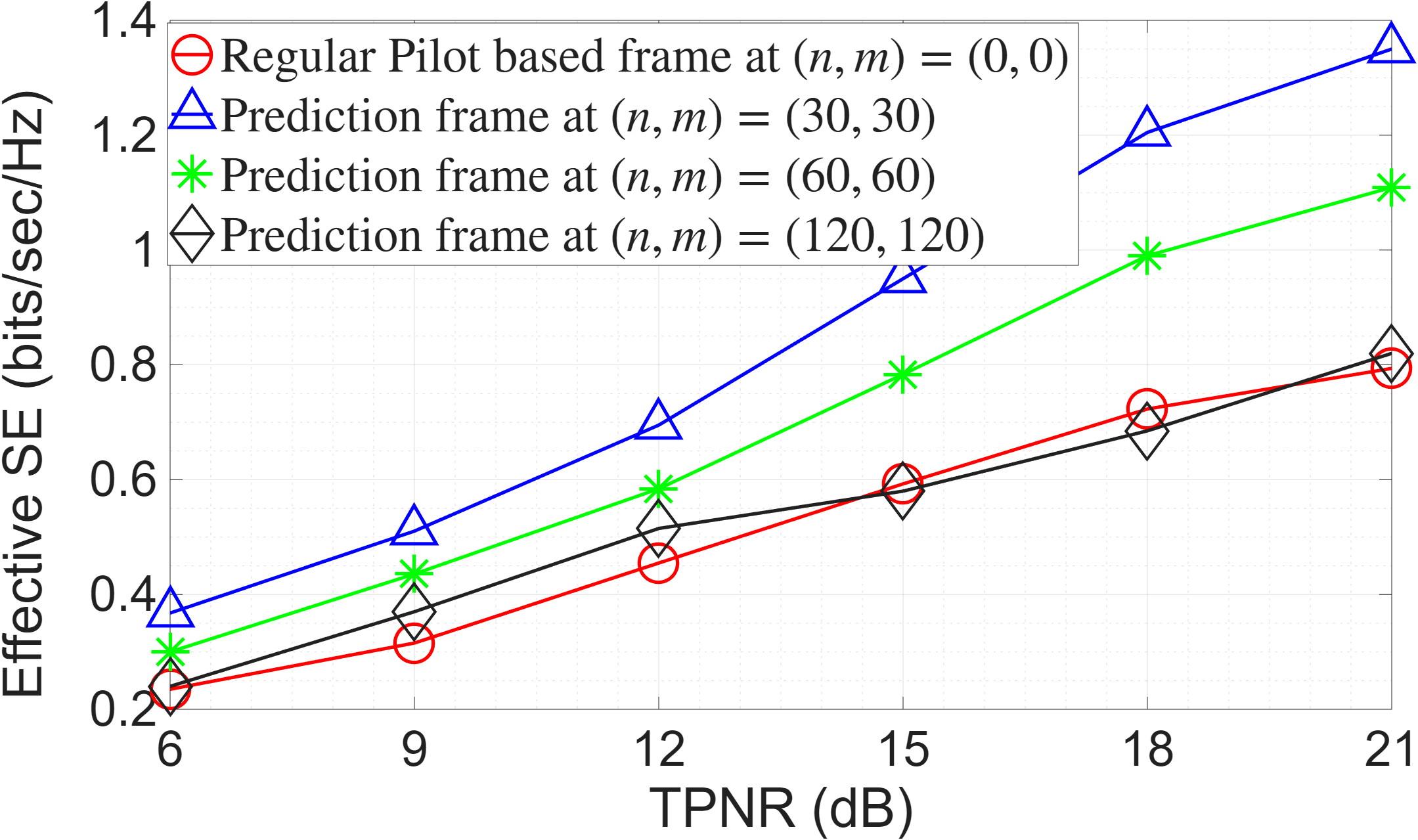}
        \hspace{-3mm}
    \caption{Effective Spectral Efficiency (bits/s/Hz) vs TPNR (in dB). Fixed $\nu_{max} = 1$ kHz. Other parameters same as Fig.~\ref{figcrlb}. $Q=30$. PDR is $0$ dB.}
    \vspace{-3mm}
\label{SEvsTPNR}
        \vspace{-3mm}
    \end{figure}    

In Fig.~\ref{SEvsTPNR} we compare the effective spectral efficiency (SE)
achieved in a traditional non-prediction based Zak-OTFS frame where a dedicated pilot is used to acquire the DD domain effective channel filter,
to that achieved in a frame where no pilot is transmitted but for which the effective DD channel filter is predicted using the proposed prediction method. The Zak-OTFS DD frame structure for the traditional frame is shown in Fig.~\ref{fig_singlepilot}. At the transmitter, the information bits are coded using a Low Density Parity Check (LDPC) code. The coded symbols are then mapped to QAM symbols (i.e., $x_{n,m}[k,l]$). The effective SE is $N_I(1 - \text{BLER})/(B'T')$ bits/sec/Hz where $N_I$ is the number of information bits communicated in the Zak-OTFS frame and BLER is the average block error rate.
The number of information bits $N_I$ is smaller in traditional frames when compared to prediction frames due to the pilot and guard region overhead in traditional frames. We choose the modulation and coding scheme (MCS) (i.e., LDPC code rate and QAM order) based on the 3GPP 5G NR standard \cite{3gppmcs1, 3gppmcs2}. To be precise, for each TPNR we choose the highest MCS for which the achieved BLER is less than or equal to $0.1$.

In Fig.~\ref{SEvsTPNR}, we observe that even for a prediction frame $60T' = 60$ ms in future and $60B' = 216$ MHz away, we achieve a $30$ percent improvement in effective SE when compared to a traditional pilot-based frame, for a moderate TPNR of $15$ dB. This is because, with the proposed prediction method the prediction frames need not carry a dedicated pilot and therefore there is no pilot and guard region overhead. The SE gains are significant even for a low TPNR of $6$ dB. 

\appendices

\section{Proof of Theorem \ref{thm1}}
\label{appen_thm1}
From (\ref{eqn8}) and (\ref{eqn1212}) we have

{\small
\vspace{-4mm}
\begin{eqnarray}
\label{eqn1313}
x_{n,m,dd}^{w}(\tau , \nu ) & = & w_{n,m,tx}(\tau, \nu) \, *_{\sigma} x_{n,m,dd}(\tau , \nu ), \nonumber \\
& & \hspace{-21mm} = \delta(\tau - nT') \delta(\nu - mB') *_{\sigma} w_{0,0,tx}(\tau, \nu) *_{\sigma} x_{n,m,dd}(\tau , \nu ), \nonumber \\
& & \hspace{-21mm} = \delta(\tau - nT') \delta(\nu - mB') *_{\sigma} x_{n,m,dd}^{w0}(\tau , \nu ), \nonumber \\
x_{n,m,dd}^{w0}(\tau , \nu ) & \Define &  w_{0,0,tx}(\tau, \nu) *_{\sigma} x_{n,m,dd}(\tau , \nu ).
\end{eqnarray}\normalsize}The transmit pulse shaping filter $w_{0,0,tx}(\tau, \nu)$ used for
the $(0,0)$-th frame localizes the transmit signal to ${\mathcal I}_{0,0}$, and therefore the TD realization of $x_{n,m,dd}^{w0}(\tau , \nu )$, i.e., $x_{n,m}^{w0}(t) \Define {\mathcal Z}_t^{-1}\left( x_{n,m,dd}^{w0}(\tau , \nu ) \right)$ is localized to the TF region ${\mathcal I}_{0,0}$. Further simplifying (\ref{eqn1313}) using (\ref{eqneqn9}), we get
\begin{eqnarray}
    x_{n,m,dd}^{w}(\tau , \nu ) & \hspace{-3mm} = & \hspace{-3mm} e^{j2\pi mB' (\tau - nT')} \, x_{n,m,dd}^{w0}(\tau - nT', \nu - mB'). \nonumber \\
\end{eqnarray}From (\ref{eqn10}),
the signal transmitted in the $(n,m)$-th frame is given by
\begin{eqnarray}
    s_{n,m}(t) & = & \sqrt{\tau_p} \int\limits_{0}^{\nu_p} x_{n,m,dd}^{w}(t , \nu ) \, d\nu, \nonumber \\
    & & \hspace{-17mm} = e^{j2\pi mB' (t - nT')} \sqrt{\tau_p} \int\limits_{0}^{\nu_p} x_{n,m,dd}^{w0}(t - nT', \nu - mB') \, d\nu, \nonumber \\
    & & \hspace{-17mm} \mya e^{j2\pi mB' (t - nT')} \sqrt{\tau_p} \int\limits_{0}^{\nu_p} x_{n,m,dd}^{w0}(t - nT', \nu) d\nu, \nonumber \\
    & & \hspace{-17mm} \myb e^{j2\pi mB' (t - nT')} x_{n,m}^{w0}(t - nT'),
\end{eqnarray}where step (a) follows from the fact that the integral in the RHS is over one Doppler period and $x_{dd}(t,\nu)$ is periodic along the Doppler axis with period $\nu_p$. Step (b) follows from the inverse Zak transform in (\ref{eqn10}). Since $x_{n,m}^{w0}(t)$ is localized to ${\mathcal I}_{0,0}$ it follows that $e^{j2\pi mB' (t - nT')} x_{n,m}^{w0}(t - nT')$ is localized to ${\mathcal I}_{n,m}$.

\section{Proof of Theorem \ref{thm2}}
\label{appen_thm2}
From (\ref{eqn2}) and (\ref{eqneqn16}) it follows that the AWGN free component of the received signal $r(t)$ is
\begin{eqnarray}
\sum\limits_{n,m \in {\mathbb Z}} \iint h_{\text{phy}}(\tau', \nu') \, s_{n,m}(t - \tau') \, e^{j2\pi\nu'(t - \tau')} d\tau' \, d\nu'. \nonumber \\
\end{eqnarray}Taking its Zak transform we get
\begin{eqnarray}
{\mathcal Z}_t\left( \sum\limits_{n,m \in {\mathbb Z}} \iint h_{\text{phy}}(\tau', \nu') \, s_{n,m}(t - \tau') \, e^{j2\pi\nu'(t - \tau')} d\tau' \, d\nu' \right) &  & \nonumber \\
& & \hspace{-95mm} = \hspace{-2mm} \sum\limits_{p,q \in {\mathbb Z}} \iint h_{\text{phy}}(\tau', \nu') \, {\mathcal Z}_t\left( s_{p,q}(t - \tau') \, e^{j2\pi\nu'(t - \tau')}\right) d\tau' \, d\nu', \nonumber \\
& & \hspace{-95mm} \mya \hspace{-2mm} \sum\limits_{p,q \in {\mathbb Z}} \hspace{-1mm} \iint \hspace{-2mm} h_{\text{phy}}(\tau', \nu') \, e^{j2\pi\nu'(\tau - \tau')} \, x_{p,q,dd}^w(\tau - \tau', \nu - \nu') d\tau' \, d\nu', \nonumber \\
& & \hspace{-95mm} \myb \hspace{-2mm} \sum\limits_{p,q \in {\mathbb Z}}  h_{\text{phy}}(\tau, \nu) *_{\sigma}  x_{p,q,dd}^w(\tau, \nu), \nonumber \\
& & \hspace{-95mm} \myc \hspace{-2mm} \sum\limits_{p,q \in {\mathbb Z}}  h_{\text{phy}}(\tau, \nu) *_{\sigma}  w_{p,q,tx}(\tau, \nu) *_{\sigma} x_{p,q,dd}(\tau, \nu),
\end{eqnarray}where step (a) follows from the fact that for any signal $a(t)$ having Zak transform $a_{dd}(\tau, \nu) = {\mathcal Z}_t(a(t))$, the DD representation of $e^{j 2 \pi \nu' (t - \tau')} a(t - \tau')$ is $e^{j 2\pi \nu' (\tau - \tau')} \, a_{dd}(\tau - \tau' , \nu - \nu')$ (see \cite{DerivOTFS,otfsbook} for details).
Step (b) follows from (\ref{eqneqn9}).
Step (c) follows from (\ref{eqn8}).
From (\ref{eqnref19}) it follows that

{\small
\vspace{-4mm}
\begin{eqnarray}
\label{eqneqn27}
    y_{n,m,dd}^w(\tau, \nu) & = & w_{n,m,rx}(\tau, \nu) *_{\sigma} {\mathcal Z}_t(r(t)), \nonumber \\
    & & \hspace{-25mm} \mya \hspace{-2mm} \sum\limits_{p,q \in {\mathbb Z}}  \underbrace{w_{n,m,rx}(\tau, \nu) *_{\sigma} h_{\text{phy}}(\tau, \nu) *_{\sigma}  w_{p,q,tx}(\tau, \nu)}_{=h_{n,m,p,q}(\tau, \nu)} *_{\sigma} x_{p,q,dd}(\tau, \nu), \nonumber \\
    & & + \, w_{n,m,rx}(\tau,\nu) *_{\sigma} z_{dd}(\tau, \nu), \nonumber \\
    & & \hspace{-25mm} = h_{n,m,n,m}(\tau, \nu) *_{\sigma} \, x_{n,m,dd}(\tau, \nu) \, + \, w_{n,m,rx}(\tau,\nu) *_{\sigma} z_{dd}(\tau, \nu) \nonumber \\
    & & + \hspace{-4mm} \sum\limits_{\substack{(p,q) \neq (n,m) \\
    p,q \in {\mathbb Z}} } \hspace{-4mm}h_{n,m,p,q}(\tau,\nu) *_{\sigma} x_{p,q,dd}(\tau, \nu),
\end{eqnarray}\normalsize}where step (a) follows from (\ref{eqneqn24}). In step (a), $z_{dd}(\tau, \nu) = {\mathcal Z}_t(z(t))$.
We next use the following result from \cite{otfsbook,zakotfs1} which is as follows. Given $b_{dd}(\tau,\nu) = g(\tau,\nu) *_{\sigma} s_{dd}(\tau,\nu)$ where $s_{dd}(\tau,\nu)= \sum\limits_{k,l \in {\mathbb Z}} s_{dd}[k,l]\delta(\tau - k/B)\delta(\nu - l/T)$ and $b_{dd}(\tau,\nu)$ are continuous
quasi-periodic DD domain functions with delay and Doppler period $\tau_p$ and $\nu_p$ respectively, then after sampling on $\Lambda$ we get $b_{dd}[k,l] = b_{dd}(\tau = k/B, \nu = l/T) = g[k,l] *_{\sigma} s_{dd}[k,l]= \sum\limits_{k',l' \in {\mathbb Z}} g[k',l'] s_{dd}[k - k', l - l'] \, e^{j 2 \pi l'(k - k')/(MN)}$, $g[k,l] = g(\tau = k/B, \nu = l/T)$. Using this result in (\ref{eqneqn27}) gives (\ref{thm2eqn1})
which completes the proof.


\section{Proof of Theorem \ref{thm3}}
\label{appen_thm3}
From \cite{zakotfs3} we know that
if $a(\tau,\nu) = b(\tau,\nu) *_{\sigma} c(\tau, \nu)$, then the matched filter of $a(\tau,\nu)$ is $a_{\text{mf}}(\tau,\nu) = a^*(-\tau, -\nu) \, e^{j 2 \pi \nu \tau} = c_{\text{mf}}(\tau, \nu) *_{\sigma} b_{\text{mf}}(\tau, \nu)$.
Using this fact in (\ref{eqn1212}),
the received matched filter $w_{n,m,rx}(\tau, \nu)$ in (\ref{eqne1818}) is given by
\begin{eqnarray}
\label{eqn29749}
    w_{n,m,rx}(\tau,\nu) = {\Big (}  w_{0,0,tx}^*(-\tau, -\nu) e^{j 2 \pi \nu \tau} *_{\sigma} & &  \nonumber \\
    & & \hspace{-43mm} \delta(\tau + nT')\delta(\nu + mB') \, e^{j2 \pi n m B' T'} {\Big )},
\end{eqnarray}where we have used the fact that the matched filter for $\delta(\tau - nT')\delta(\nu - mB')$ is simply $\delta(\tau + nT')\delta(\nu + mB') \, e^{j2 \pi n m B' T'}$.
In (\ref{eqne1818}),
the receive matched filter for the $(0,0)$-th frame is
\begin{eqnarray}
    \label{eqneqn12478}
    w_{0,0,rx}(\tau,\nu) & = & w^*_{0,0,tx}(-\tau, -\nu) \, e^{j 2\pi \nu \tau}.
\end{eqnarray}Using (\ref{eqneqn12478}) in (\ref{eqn29749}) we finally get
\begin{eqnarray}
\label{eiq7739}
    w_{n,m,rx}(\tau, \nu) & = & w_{0,0,rx}(\tau,\nu) *_{\sigma} \nonumber \\
    & & \delta(\tau + nT')\delta(\nu + mB') \, e^{j2 \pi n m B' T'}.
\end{eqnarray}
Using (\ref{eqn1212}) and (\ref{eiq7739}) in (\ref{eqneqn24}) it follows that
\begin{eqnarray}
\label{eqn395696}
    h_{n,m,p,q}(\tau, \nu) & = & w_{n,m,rx}(\tau,\nu) *_{\sigma} h_{\text{phy}}(\tau, \nu) *_{\sigma} w_{p,q,tx}(\tau,\nu), \nonumber \\
    & & \hspace{-20mm}=  w_{0,0,rx}(\tau, \nu) *_{\sigma} h_{n,m,p,q, \text{phy}}(\tau, \nu) *_{\sigma} w_{0,0,tx}(\tau, \nu),
    \end{eqnarray}
\begin{eqnarray}
\label{eqn297495}
    h_{n,m,p,q,\text{phy}}(\tau, \nu) & \Define & \left( \delta(\tau + nT')\delta(\nu + mB') \, e^{j2 \pi n m B' T'} \right) \nonumber \\
    & & \hspace{-10mm} *_{\sigma} \, h_{\text{phy}}(\tau, \nu) \, *_{\sigma} {\Big (} \delta(\tau - pT') \delta(\nu - qB') {\Big )}, \nonumber \\
    & & \hspace{-30mm} = {\Big (} e^{j 2 \pi p(m - q) B' T'} \, e^{j 2 \pi \left( \nu p T' - \tau m B' \right)}  \nonumber \\
    & & \hspace{-20mm} h_{\text{phy}}(\tau + (n - p)T', \nu + (m-q)B') {\Big )},
\end{eqnarray}where the last step follows from (\ref{eqneqn9}).
Using the expression for $h_{\text{phy}}(\tau, \nu)$ in (\ref{eqn1}) we get the expression for $h_{n,m,p,q,\text{phy}}(\tau, \nu)$ in terms of $h_i,\tau_i, \nu_i$, $i=1,2,\cdots,P$, in (\ref{eqn94u03850}) (see top of next page).
\begin{figure*}
\vspace{-8mm}
\begin{eqnarray}
\label{eqn94u03850}
    h_{n,m,p,q,\text{phy}}(\tau, \nu) & = &   e^{j 2 \pi m(n-p)B'T'} \sum\limits_{i=1}^P h_i \, e^{j2 \pi \left(pT'\nu_i - mB'\tau_i \right)} \, \delta(\tau - \tau_i - (p -n)T') \, \delta(\nu - \nu_i - (q - m)B').
\end{eqnarray}
\vspace{-4mm}
\begin{eqnarray*}
    \hline
\end{eqnarray*}
\end{figure*}For $(p,q) = (n,m)$, from (\ref{eqn94u03850}) we get
\begin{eqnarray}
\label{eqn204653}
    h_{n,m,n,m,\text{phy}}(\tau, \nu)  & \hspace{-2mm} = & \hspace{-2mm} \sum\limits_{i=1}^P h_i \alpha_i^n \beta_i^m \delta(\tau - \tau_i) \delta(\nu - \nu_i),
\end{eqnarray}where $\alpha_i, \beta_i$, $i=1,2,\cdots,P$ is given by (\ref{eqnthm3}). Using (\ref{eqn204653}) in (\ref{eqn395696}) we get
\begin{eqnarray}
\label{eqn924950}
    h_{n,m,n,m}(\tau,\nu) & = & \sum\limits_{i=1}^P  \alpha_i^n  \, \beta_i^m \, A_i(\tau, \nu),
\end{eqnarray}where $A_i(\tau, \nu)$ is given by (\ref{eqnthm3}). From (\ref{thm2eqn3}) and (\ref{eqn924950})
it follows that
\begin{eqnarray}
    h_{n,m}[k,l] & = & h_{n,m,n,m}\left( \tau = \frac{k}{B} , \nu = \frac{l}{T}\right), \nonumber \\
    & = & \sum\limits_{i=1}^P \alpha_i^n \beta_i^m A_i[k,l],
\end{eqnarray}where $A_i[k,l]$ is given by (\ref{eqnthm3}).

\section{Proof of Lemma \ref{lem1}}
\label{appen_lem1}
The factorization follows directly from the expression of $h_{n,m}[k,l]$ in (\ref{eqnthm3}) of Theorem \ref{thm3}.
The time-phase matrix $\Phi_{q_1,q_2}$ in (\ref{eqn37}) can be factored as

{\vspace{-4mm}
\small
\begin{eqnarray}
\label{eqn28648}
     \Phi_{q_1, q_2} & & \nonumber \\
     & & \hspace{-14mm} = \begin{bmatrix}
         \alpha_1^{-q_2} & 0 & \dots & 0 \\
         0 & \alpha_2^{-q_2} & \dots & 0 \\
         \vdots & \vdots & \vdots & \vdots \\
         0 & 0 & \dots & \alpha_P^{-q_2} \, 
     \end{bmatrix} 
     \begin{bmatrix}
         \alpha_1^{-(q_1-q_2)} &  \dots &  \alpha_1^{-1} &  1 \\
          \alpha_2^{-(q_1-q_2)} &  \dots &  \alpha_2^{-1} &  1\\
        \vdots & \vdots & \vdots \\
          \alpha_P^{-(q_1-q_2)} &  \dots &  \alpha_P^{-1} &  1\\
    \end{bmatrix}. \nonumber \\
\end{eqnarray}\normalsize}
The R.H.S. in (\ref{eqn28648}) is a product of two matrices, a full rank square diagonal $P \times P$ matrix and a Vandermonde matrix \cite{Golub, Horn13}. The Vandermonde matrix has full row rank $P$ since all $\alpha_i, i=1,2,\cdots, P$ are distinct (see Assumption $4$) and $P \leq (q_1 - q_2 + 1)$ (see (\ref{rankcondition})). This therefore proves that $\Phi_{q_1, q_2}$ has full row rank $P$.
From Assumption $2$ and $3$ we know that $N_t \geq P$ and ${\bf A}$ has full column rank $P$.
Since ${\bf H}_{t,q_1,q_2} = {\bf A} \Phi_{q_1, q_2}$, it therefore follows that $H_{t,q_1,q_2}$ has rank $P$.


\end{document}